\documentclass[11pt]{article}
\usepackage[letterpaper, margin=1in]{geometry}

\usepackage{graphicx} 
\usepackage{amsmath,amssymb,amsfonts,amsthm,mathtools,bm}
\usepackage{aliascnt}
\usepackage{braket}
\usepackage{booktabs}
\usepackage{tabularx}
\usepackage{float}
\usepackage{algorithm}
\usepackage[noend]{algpseudocode}

\usepackage{enumitem}
\usepackage{microtype}
\usepackage{xcolor}
\usepackage{hyperref}
\definecolor{navy}{RGB}{0,90,200}
\definecolor{burgundy}{RGB}{128,0,32}
\hypersetup{colorlinks=true, linkcolor=navy, citecolor=burgundy, urlcolor=navy}
\usepackage{comment}

\newcommand{\om}{\omega}

\theoremstyle{plain}
\newtheorem{theorem}{Theorem}[section]
\newaliascnt{lemma}{theorem}
\newtheorem{lemma}[lemma]{Lemma}
\aliascntresetthe{lemma}
\newaliascnt{proposition}{theorem}
\newtheorem{proposition}[proposition]{Proposition}
\aliascntresetthe{proposition}
\newaliascnt{corollary}{theorem}
\newtheorem{corollary}[corollary]{Corollary}
\aliascntresetthe{corollary}
\newaliascnt{assumption}{theorem}

\aliascntresetthe{assumption}
\theoremstyle{definition}
\newaliascnt{definition}{theorem}
\newtheorem{definition}[definition]{Definition}
\aliascntresetthe{definition}
\theoremstyle{remark}
\newaliascnt{remark}{theorem}
\newtheorem{remark}[remark]{Remark}
\aliascntresetthe{remark}

\usepackage[nameinlink,capitalise]{cleveref}
\Crefname{theorem}{Theorem}{Theorems}
\Crefname{lemma}{Lemma}{Lemmas}
\Crefname{proposition}{Proposition}{Propositions}
\Crefname{corollary}{Corollary}{Corollaries}
\Crefname{assumption}{Assumption}{Assumptions}
\Crefname{definition}{Definition}{Definitions}
\Crefname{remark}{Remark}{Remarks}
\Crefname{equation}{Eq.}{Eqs.}
\Crefname{claim}{Claim}{Claims}

\newcommand{\C}{\mathbb{C}}
\newcommand{\R}{\mathbb{R}}

\newcommand{\diag}{\operatorname{diag}}

\newcommand{\NUCT}{\textsc{nuct}}

\newcommand{\NUQFT}{\textsc{nuqft}}

\newcommand{\PREP}{\operatorname{PREP}}
\newcommand{\SEL}{\operatorname{SEL}}

\newcommand{\nint}{\operatorname{round}}

\newcommand{\eps}{\varepsilon}
\newcommand{\ketzero}[1]{\ket{0}^{\otimes #1}}
\newcommand{\norm}[1]{\left\lVert #1\right\rVert}
\newcommand{\abs}[1]{\left\lvert #1\right\rvert}
\newcommand{\setof}[1]{\left\{#1\right\}}

\newcommand{\matU}{\bold U}

\newcommand{\Z}{{\mathbb Z}}

\title{Conditioning-Free Non-Uniform Quantum Fourier and Chebyshev Transforms}

\author{
  Chaowen Guan\textsuperscript{1} and Akshit Katiyar\textsuperscript{2}\\[1ex]
  \small\itshape \textsuperscript{1}Department of Computer Science, University of Cincinnati, OH\\
  \small\itshape \textsuperscript{2}Department of Computer Science, Pennsylvania State University, University Park, PA
}
\date{ }

\begin{document}

\maketitle
\footnotetext[1]{\texttt{guance@ucmail.uc.edu}}
\footnotetext[2]{\texttt{akshitk@psu.edu}}

\begin{abstract}
We present an exponentially fast quantum algorithm for the non-uniform Chebyshev transform. It is defined as the projection of a function onto Chebyshev polynomials sampled at given nodes that are uniform in $x\in[-1,1]$, and hence non-uniform in the angle $\theta=\arccos x$, a setting that QFT-based quantum Chebyshev transforms cannot handle. Our construction is based on an improvement of the Non-uniform Quantum Fourier Transform (NUQFT) whereby we remove the conditioning from non-uniform node sampling. Hence, error bounds are independent of the geometry-dependent parameter $\kappa$ of prior work. We use the fact that Chebyshev transform matrix is the average of two Type-II non-uniform DFTs, which we implement with a single controlled NUQFT circuit. We provide explicit oracle constructions, including the row-access oracle previously left as an assumption. The resulting $\varepsilon$-accurate block encoding has $O(1)$ normalization and uses $O(L)$ qubits and $\widetilde O(L^2)$ gates, where $L=\log N+\log(1/\varepsilon)$. We give an end-to-end complexity analysis, including the success probability and output-state error.
\end{abstract}

\tableofcontents

\clearpage

\section{Introduction}

Chebyshev expansions are a standard tool in numerical methods for approximating functions on an interval.  For smooth functions the Chebyshev coefficients decay rapidly, so a few terms give high accuracy.  This underlies spectral methods for differential equations and numerical software for computing with functions \cite{Tow14}.  In quantum computing, Chebyshev expansions have been used to embed, learn and sample from distributions \cite{WPWEK23}.  In each case the Chebyshev transform is the step that converts sampled function values into coefficients of expansion.

Chebyshev polynomials connect polynomial approximation on an interval with Fourier analysis through the identity
\[
T_j(x)=\cos(j\theta),\qquad x=\cos\theta,
\]
which turns a Chebyshev expansion in \(x\in[-1,1]\) into a cosine series in \(\theta\).  At Chebyshev nodes \(x_k=\cos\theta_k\) with uniformly spaced angles, the discrete Chebyshev transform is therefore a discrete cosine transform.  However, when the sample points are evenly spaced in \(x\), their angles \(\theta_k=\arccos x_k\) are unevenly spaced, and the cosine transform no longer applies directly.  Classically, the cosine transform is a discrete Fourier transform of an even extension of the input \cite{Mak80}, and quantum cosine and sine transforms follow from the Quantum Fourier Transform (QFT) \cite{KR01}.  The quantum Chebyshev transform of \cite{WPWEK23} also uses the QFT, but only at Chebyshev roots, whose uniformly spaced angles make the sampled polynomials discretely orthogonal and the transform unitary.

In this work, we consider a discrete Chebyshev transform in the broader sense of the discrete polynomial transforms of Driscoll, Healy, and Rockmore \cite{DHR97}. Given sample values \(f_k\), nodes \(x_k\), and weights \(w_k\), this transform computes
\[
c_j=\sum_{k=0}^{N-1}w_k f_k T_j(x_k),
\]
where the nodes can be arbitrary. We call the resulting operation the non-uniform Chebyshev transform (NUCT) because it is non-uniform compared to standard definitions of Chebyshev transforms. Specifically, this work consider nodes evenly spaced in \(x\), which become non-uniformly spaced after the change of variables \(\theta_k=\arccos x_k\). However, our algorithm can be easily extended for arbitrarily spaced grid points assuming access to the required oracles. Its outputs are the $\ell^2$ projections of the data onto the sampled Chebyshev polynomials. They are not interpolation coefficients, because the Chebyshev polynomials are not orthogonal on an arbitrary grids.

Our original motivation was to develop a quantum implementation of the \(O(N\log^2N)\) polynomial-transform framework of \cite{DHR97}, which exploits three-term recurrences.  The NUCT is exactly the first stage of that framework.  Although not orthogonal on the uniform grid, the Chebyshev polynomials are much better conditioned there than the monomials, so their projections are a stable starting point.  Three-term recurrences then convert them into projections onto other orthogonal families such as the Hahn polynomials, and more generally onto any polynomial family that satisfies a three-term recurrence.  Here we focus solely on the NUCT and leave the extension to more general polynomial families for future work. 

We write the NUCT as an average of two non-uniform Fourier sums,
\[
c_j=\frac12\sum_{k=0}^{N-1}f_k
\left(e^{ij\theta_k}+e^{-ij\theta_k}\right).
\]
Since this operator is generally non-unitary, its quantum realization is formulated through a block encoding, rather than an orthonormal change of basis on the input register. We use the non-uniform quantum Fourier transform (NUQFT) of \cite{AKY26} to implement the corresponding non-uniform Chebyshev operator. Tian and Liu \cite{TL00} did it classically with an oversampled least-squares interpolation NUFFT. \cite{AKY26} instead used the low-rank NUFFT of \cite{RT18}, which needs no oversampling and admits explicit error bounds.

However, the error and gate bounds of the NUQFT in \cite{AKY26} depend on a parameter $\kappa$, which grows without bound as any scaled node approaches the midpoint between two grid points.  It measures the conditioning of an auxiliary $\arccos$ evaluation rather than of the transform itself.  For the NUCT nodes we could not find a useful bound on $\kappa$ (\Cref{app:kappa-nuct}).  This raises the question of whether the NUQFT can be made to work with parameters that do not depend on the node geometry.  We answer it affirmatively by reworking the analysis of \cite{AKY26} so that $\kappa$ is eliminated, with better .

\paragraph{Problem Statement.}
\label{sec:problem-statement}
We seek a quantum algorithm that takes samples of a function at equispaced nodes, encoded in a quantum state, and outputs their NUCT as a normalized quantum state.

\begin{definition}[Non-uniform Chebyshev transform]
\label{def:nuct}
Let $N=2^n$ with $n\ge2$, and let $x_k=-1+\frac{2k}{N}$, $k=0,\ldots,N-1$, be equispaced nodes in $[-1,1]$.  The non-uniform Chebyshev transform (NUCT) maps $f\in\C^N$ to
\begin{equation}
  c_j=\sum_{k=0}^{N-1}f_kT_j(x_k),
  \qquad j=0,\ldots,N-1.
  \label{eq:nupct-def}
\end{equation}
\end{definition}
We also define the normalized transform matrix
\begin{equation}
  (C_N)_{jk}=\frac{1}{\sqrt N}T_j(x_k),
  \label{eq:C-matrices}
\end{equation}
so that $c=\sqrt N\,C_Nf$.

Given a circuit that prepares
$\ket f=\norm f_2^{-1}\sum_kf_k\ket k$, the goal is to prepare an
approximation of
\begin{equation}
  \frac{C_N\ket f}{\norm{C_N\ket f}}=\frac{c}{\norm c_2}.
  \label{eq:quantum-task}
\end{equation}
We construct an $\eps$-accurate block encoding of $C_N$.  The state then
follows by postselection and amplitude amplification \cite{BHMT02}, with an error that
depends on $\eps$ and $\norm{C_N\ket f}$ (\Cref{cor:output-state}).

\subsection{Main contributions}
\label{sec:intro-main-contributions}

\paragraph{$\kappa$ independent block-encoding analysis.} The low-rank Chebyshev-Bessel expansion and the overall circuit structure are inherited from \cite{RT18,AKY26}. We obtain a simpler truncation bound, error bounds that do not depend on $\kappa$, and an improved block-encoding normalization. As in \cite{RT18}, several nodes may round to the same grid point. Unlike the classical algorithm, where such collisions are free, they enter our block encoding through its normalization $O(\sqrt{d_r})$, where $d_r$ is the largest number of nodes sharing a grid point (\Cref{lem:routing-factorization}). Our general NUQFT is therefore efficient when $d_r=\mathrm{polylog}(N)$.  In the worst case $d_r=N$, and the normalization grows to $\sqrt N$.  Our NUQFT result is summarized below.
\begin{theorem}[Informal version of \Cref{thm:dyadic-nuqft}]
Given row access to $N$ arbitrary nodes, at most $d_r$ of which share a grid point, the NUDFT $F_t$ has an $\eps$-accurate block encoding with normalization $O(\sqrt{d_r})$ and $\widetilde O(\log^2(N/\eps))$ gates.  The cost does not depend on $\kappa$.
\end{theorem}

\paragraph{Application to the NUCT.} Combining the reduction of the NUCT to two Type-II NUDFTs (\Cref{thm:nupct-reduction}) with our conditioning-free NUQFT gives an $\varepsilon$-accurate block encoding of the NUCT without any geometry-dependent parameter such as $\kappa$.  For these nodes we show that at most five angles round to the same grid point, so $d_r\le5$, and we construct the row-access oracle explicitly (\Cref{prop:row-oracle}).  The result therefore needs no assumption on collisions or on row access. The main result can be summarized in the theorem below:
\begin{theorem}[Informal version of \Cref{thm:compiled-nupct,cor:output-state}]
\label{thm:informal-nuct}
The NUCT on $N$ equispaced nodes has an explicit $\eps$-accurate block encoding with normalization $O(1)$, using $O(\log(N/\eps))$ qubits and $\widetilde O(\log^2(N/\eps))$ gates.  It prepares $C_N\ket f/\norm{C_N\ket f}$ with $O(1/\norm{C_N\ket f})$ uses of this circuit.
\end{theorem}
The transform itself thus costs only polylogarithmically many gates in $N$ and $1/\eps$.  Preparing the output state multiplies this, together with the cost of preparing $\ket f$, by $O(1/\norm{C_N\ket f})$, a factor set by the input rather than by the node geometry.

\paragraph{Arbitrary nodes.} Our technique for NUCT can be extended to arbitrary nodes. More concretely, the proof of \Cref{thm:informal-nuct} uses only $x_k=\cos\theta_k$, so the reduction holds for any nodes $x_k\in[-1,1]$. Equispacing is used only to build the node oracle, to bound $d_r\le5$ and to construct the row oracle. For other nodes, \Cref{thm:dyadic-nuqft} gives the same construction with normalization $O(\sqrt{d_r})$, given a node oracle with $\norm{t-\tau}_\infty\le2^{-m}$ and row access for those nodes.

The table below summarizes our improvements for application to NUCT 
\begin{table}[h]
\centering
\small
\begin{tabular}{@{}lll@{}}
\toprule
 & \NUQFT{} of \cite{AKY26} applied to the \NUCT & This work \\
\midrule
Dependence on node geometry & $\log(1+\kappa L)$, best proven bound $O(N\log N)$ & none \\
Block-encoding normalization & $O(K^2\sqrt{d_r})=O(L^2)$ & $\sqrt{d_r}\,\Lambda<24\sqrt5$ \\
Collision bound $d_r$ & assumed & $d_r\le5$, proven \\
Row-access and coefficient oracles & assumed & explicit circuits \\
\bottomrule
\end{tabular}
\caption{Here $L=\log N+\log(1/\eps)$, $K=O(L)$ is the truncation rank, $d_r$ bounds the number of nodes that round to the same grid point, and $\Lambda$ is the total Bessel weight.}
\label{tab:comparison}
\end{table}

\paragraph{Organization.}  \Cref{sec:prelim} introduces the preliminaries like a Bessel function bound and the quantum circuits for trigonometric functions used throughout.  \Cref{sec:background} defines the Type-II NUDFT recalls the important lemmas of \cite{RT18} and \cite{AKY26}, showing where $\kappa$ enters the analysis and lists the issues we address.  \Cref{sec:CFNUQFT} provides a sharper analysis removing $\kappa$ from the NUQFT algorithm in \Cref{thm:dyadic-nuqft}.  \Cref{sec:nupct} applies it to the NUCT by using the reduction to two NUDFTs and constructing explicit node, row-access and coefficient-preparation oracles along with end-to-end complexity in \Cref{thm:compiled-nupct,cor:output-state}.

\section{Preliminaries}
\label{sec:prelim}

\noindent \textbf{Fixed-point number.} For an integer $m\ge1$, an $m$-bit fixed-point number in $[0,1)$ is
any number of the form
\[
  \tau=\frac{a}{2^m},
  \qquad
  a\in\{0,1,\ldots,2^m-1\}.
\]
Equivalently, $\tau$ has a terminating binary expansion $\tau=0.b_1b_2\cdots b_m$.  Stored nodes are written
\begin{equation}
  \tau_j=\frac{a_j}{2^m},
  \qquad
  a_j\in\{0,\ldots,2^m-1\}.
  \label{eq:fixed-point-node}
\end{equation}  
Offsets such as $z(t)\in[-1,1]$ can be negative, so we also use \emph{signed} fixed-point numbers.  A signed fixed-point number with $p$ fractional bits is any number $a/2^p$ with $a\in\Z$, and a register stores it as the integer $a$ in two's complement.  Shifts, additions, subtractions and comparisons of fixed-point numbers are exact.  We call a quantity \emph{exact} when it is a fixed-point number held in a register with no rounding.

\begin{remark}
    When a quantum register stores the integer $a$, we regard it as an exact representation of the mathematical value $a/2^m$.  A real or irrational node $t$ may first be approximated by a fixed-point node $\tau$. The fixed-point NUQFT is then applied exactly to $\tau$ and the matrix perturbation from $t$ to $\tau$ is accounted for separately.
\end{remark}

\noindent \textbf{Nearest grid point and offset.}  Let $N=2^q$.  For a node $t\in[0,1)$ define
\begin{align}
  s^{\mathrm{unw}}(t)
  &:=\nint(Nt)\in\{0,\ldots,N\}, \notag\\
  \sigma(t)
  &:=s^{\mathrm{unw}}(t)\bmod N\in\{0,\ldots,N-1\}, \notag\\
  z(t)
  &:=2\bigl(Nt-s^{\mathrm{unw}}(t)\bigr)\in[-1,1],
  \label{eq:dyadic-routing-data-a}
\end{align}
where $\nint(\cdot)$ rounds to the nearest integer.  Here $s^{\mathrm{unw}}(t)/N$ is the grid point nearest to $t$, and $z(t)$ is the distance from $t$ to it, scaled to $[-1,1]$.  The index $\sigma(t)$ is the same grid point wrapped around to $\{0,\ldots,N-1\}$, so that it names a column of the $N\times N$ DFT.  The two differ only when $t$ is within $1/(2N)$ of $1$, where $s^{\mathrm{unw}}(t)=N$ and $\sigma(t)=0$.  Since $Nt=s^{\mathrm{unw}}(t)+z(t)/2$ and $s^{\mathrm{unw}}(t)\equiv\sigma(t)\pmod N$, for every integer frequency $k$,
\begin{equation}
  e^{-2\pi ikt}
  =e^{-2\pi ik\sigma(t)/N}\,e^{-i\pi(k/N)z(t)}.
  \label{eq:background-phase-split}
\end{equation}
The first factor is an entry of the uniform DFT.  The second depends on the node only through the offset $z(t)$.  For a node sequence we write $s_j^{\mathrm{unw}}=s^{\mathrm{unw}}(t_j)$, $\sigma_j=\sigma(t_j)$ and
\begin{equation}
  M_\sigma:=\sum_{j=0}^{N-1}\ket{\sigma_j}\!\bra{j},
  \qquad
  d_\ell:=\bigl|\{j:\sigma_j=\ell\}\bigr|,
  \qquad
  d_\tau:=\max_{0\le\ell<N}d_\ell,
  \label{eq:prelim-routing-matrix-data}
\end{equation}
so that $M_\sigma$ sends node $j$ to its grid point, $d_\ell$ is the number of nodes assigned to grid point $\ell$, and $d_\tau$ is the largest such number.  For real nodes $t_j$ (\Cref{sec:background}) we write the offset as $y_j=z(t_j)$, following \cite{AKY26}.  For the stored fixed-point nodes $\tau_j$ of \Cref{eq:fixed-point-node} (\Cref{sec:CFNUQFT}) we write $z_j=z(\tau_j)$.


\paragraph{Bessel functions.} We use the Bessel functions of the first kind $J_n$ of integer order, together with the symmetries \cite[Eqs.~10.4.1 and 10.11.1]{OLBC10}
\begin{equation}
  J_{-n}(a)=(-1)^nJ_n(a),
  \qquad
  J_n(-a)=(-1)^nJ_n(a),
  \label{eq:bessel-symmetries}
\end{equation}
so in particular $|J_{-n}(a)|=|J_n(-a)|=|J_n(a)|$.  Bessel functions enter the paper through the Jacobi--Anger expansion.  Setting $t=ie^{iu}$ in the generating function $e^{z(t-1/t)/2}=\sum_{n\in\Z}t^nJ_n(z)$ \cite[Eq.~10.12.1]{OLBC10} gives $e^{iz\cos u}=\sum_ni^nJ_n(z)e^{inu}$.  With $z=-a$ and \Cref{eq:bessel-symmetries},
\begin{equation}
  e^{-ia\cos u}=\sum_{n\in\Z}(-i)^nJ_n(a)\,e^{inu},
  \qquad a\in\R,\ u\in\R.
  \label{eq:jacobi-anger}
\end{equation}
The following bound gives the super-geometric decay of $J_n(a)$ in $n$.

\begin{lemma}\label{lm:bessel_bound}
    For $n \in \Z_{\geq 0}$ and $a \geq 0$, 
    \[
    |J_n(a)| \leq e^{a^2/4}\frac{(a/2)^n}{n!}.
    \]
\end{lemma}
\begin{proof}
    For every $n\in\mathbb Z_{\ge0}$, the defining series for the Bessel
function of the first kind \cite[Eq.~10.2.2]{OLBC10} gives
\begin{equation*}
  J_n(a)
  =
  \left(\frac a2\right)^n
  \sum_{k=0}^{\infty}
  \frac{(-1)^k(a^2/4)^k}{k!(n+k)!}.
\end{equation*}
Consequently, using $(n+k)!\ge n!$,
\begin{align*}
  |J_n(a)|
  &\le
  \left(\frac a2\right)^n
  \sum_{k=0}^{\infty}
  \frac{(a^2/4)^k}{k!(n+k)!} \notag\\
  &\le
  \frac{(a/2)^n}{n!}
  \sum_{k=0}^{\infty}\frac{(a^2/4)^k}{k!}
  =
  e^{a^2/4}\frac{(a/2)^n}{n!}.
\end{align*}
\end{proof}

\paragraph{Trignonometric evaluation circuits.} We use standard reversible circuits for elementary functions on fixed-point inputs.  CORDIC evaluates these functions by a sequence of shift-and-add rotations and \cite{BBG24} gave a clean reversible quantum version of it. \Cref{thm:reversible-elementary-functions} follows from their construction. 
\begin{theorem}[Reversible inverse cosine and trigonometric evaluation]
\label{thm:reversible-elementary-functions}
Let $b,p\ge1$ and let the input be exact in signed $b$-bit fixed point.  There
are clean reversible circuits that return, with $p$ fractional bits,
\begin{enumerate}[label=(\roman*),leftmargin=2.4em]
  \item an angle $\widetilde\phi(z)\in[0,\pi]$ with
  $|\widetilde\phi(z)-\arccos z|\le c_{\rm acos}2^{-p}$ for every
  $z\in[-1,1]$, exact at $z=1$.  At $z=-1$ an endpoint flag may instead
  implement $T_n(-1)=(-1)^n$ exactly.
  \item values $\widetilde s(\theta),\widetilde c(\theta)$ with
  $|\widetilde s(\theta)-\sin\theta|+|\widetilde c(\theta)-\cos\theta|
  \le c_{\rm trig}2^{-p}$ for every $\theta$ in a fixed bounded interval,
\end{enumerate}
using
\begin{equation}
  O(b+p)\ \text{qubits},\qquad
  \widetilde O\!\left((b+p)^2\right)\ \text{gates},\qquad
  \widetilde O(b+p)\ \text{depth}.
  \label{eq:reversible-elementary-resources}
\end{equation}
\end{theorem}

\begin{proof}
For (i), run the reversible CORDIC $\arcsin$ circuit of \cite{BBG24}
($O(w)$ qubits, $O(w^2)$ gates, $O(w\log w)$ depth) at precision
$w=\Theta(b+p)$ and use $\arccos z=\pi/2-\arcsin z$, with a coherent
comparison for the endpoints.  For (ii), use reversible CORDIC in rotation
mode after range reduction \cite{HanerRoettelerSvore2018}, with the same
bounds.  In both cases, copying the output and uncomputing clears all
scratch registers.
\end{proof}

\section{Background on  Low-Rank NUFFT and NUQFT}
\label{sec:background}
Our algorithm modifies the non-uniform quantum Fourier transform (NUQFT) of \cite{AKY26}, which in turn is a quantum realization of the low-rank non-uniform FFT of \cite{RT18}.  This section defines the Type-II NUDFT \Cref{eq:type-II-mat}, recalls both constructions as they appear in those works, explains where the parameter $\kappa$ enters the cost of the NUQFT, and lists the changes that our algorithm makes in \Cref{sec:CFNUQFT} .

\paragraph{The Type-II NUDFT.}
Following \cite{AKY26}, let $\{t_k\}_{k=0}^{N-1}\subseteq[0,1)$ be possibly non-uniform nodes and $\{\om_j\}_{j=0}^{N-1}$ uniform frequencies, $\om_j=j$.  The Type-II NUDFT matrix is
\begin{equation}
  (F_t)_{jk} = \frac{1}{\sqrt{N}} e^{-2\pi i \om_j t_k}=\frac{1}{\sqrt N}e^{-2\pi i jt_k},
  \qquad 0\le j,k<N.
  \label{eq:type-II-mat}
\end{equation}
Throughout the paper, rows are indexed by frequencies and columns by nodes.  So $F_t$ takes data given at the nodes and returns values at the uniform frequencies.  This is the orientation of the NUCT \Cref{eq:nupct-def}, whose output is indexed by the degree $j$ and whose sum runs over the nodes, and of the discrete polynomial transforms of \cite{DHR97}.  The rest of \Cref{sec:background} and \Cref{sec:CFNUQFT} rename the indices and write $(F_t)_{kj}=\frac{1}{\sqrt N}e^{-2\pi ikt_j}$, with frequency index $k$ and node index $j$.  Without a subscript, $F$ denotes the uniform DFT, the case $t_k=k/N$.

\begin{remark}[Orientation convention]
\label{rem:orientation}
\cite{RT18} and \cite{AKY26} write the Type-II matrix with the opposite orientation, nodes on the rows and frequencies on the columns:
\[
  (F_{\mathrm{II}})_{kj}=e^{-2\pi it_kj},
  \qquad\text{so that}\qquad
  F_t=\frac{1}{\sqrt N}F_{\mathrm{II}}^{\mathsf T}.
\]
In their terminology the transpose $F_{\mathrm{II}}^{\mathsf T}$ is called the Type-I NUDFT.  We nevertheless call $F_t$ the Type-II matrix, because the nodes and the uniform frequencies play the same roles as in \cite{AKY26}, and only the orientation differs.  Concretely:
\begin{itemize}
  \item Every matrix in this paper, including the factorization \Cref{eq:background-matrix-low-rank}, is written for $F_t$, with frequencies on the rows.
  \item Complexity bounds quoted from \cite{AKY26} carry over unchanged, since $\norm{A^{\mathsf T}}=\norm{A}$.
  \item The matrix $F_{\mathrm{II}}$ of \cite{AKY26} is obtained by transposing our construction, as explained after \Cref{thm:dyadic-nuqft}.
\end{itemize}
\end{remark}

\subsection{Low-rank factorization of the NUDFT}
\label{sec:rt18-lowrank}

Applying the NUDFT directly costs $O(N^2)$ operations.  When the nodes are uniform, $t_j=j/N$, the FFT does the same in $O(N\log N)$.  The non-uniform FFT of \cite{RT18} reduces the general case to a few uniform FFTs.  It uses the fact that every node lies within half a grid spacing of some uniform grid point $s/N$.  Relative to that grid point, $e^{-2\pi ikt_j}$ is a uniform DFT entry times a correction factor that depends only on the frequency $k$ and on the small offset of $t_j$ from the grid point.  The correction factor is smooth, so it is well approximated by a short sum of terms, each a function of the node times a function of the frequency.  Each such term is a uniform DFT with a diagonal scaling on either side, so the NUDFT becomes a sum of a few scaled DFTs.  We recall the construction in three steps, in the notation used throughout this paper.

\paragraph{Split each node into a grid point and an offset.}
Each node $t_j$ is split into its nearest grid point $\sigma_j$ and the scaled offset $y_j=z(t_j)\in[-1,1]$ of \Cref{eq:dyadic-routing-data-a}.  By \Cref{eq:background-phase-split}, $e^{-2\pi ikt_j}=e^{-2\pi ik\sigma_j/N}e^{-i\pi(k/N)y_j}$ which is precisely an entry of the uniform DFT times a correction factor. This factor is the only part that depends on the offset and couples $y_j$ with the frequency $k$.

\paragraph{Separate the node and frequency dependence.}
To turn the correction factor into a sum of node-times-frequency terms, first rescale the frequency to $w_k:=2k/N-1\in[-1,1)$, so that both variables lie in $[-1,1]$.  Then
\begin{equation}
  e^{-i\pi(k/N)y_j}=e^{-i\pi y_j/2}\,H(w_k,y_j),
  \qquad
  H(x,z):=e^{-i\pi xz/2},\quad (x,z)\in[-1,1]^2.
  \label{eq:residual-kernel}
\end{equation}
The prefactor $e^{-i\pi y_j/2}$ depends only on the node, while the kernel $H$ couples both variables.  Because $H$ is analytic, its Chebyshev expansion in both variables converges rapidly, and each term $T_\ell(z)T_r(x)$ is a function of the node times a function of the frequency.  Writing $x=\cos\theta$, $z=\cos\varphi$ and $a=\pi/4$, one has $e^{-i\pi xz/2}=e^{-ia\cos(\theta+\varphi)}e^{-ia\cos(\theta-\varphi)}$.  Applying the Jacobi--Anger expansion \Cref{eq:jacobi-anger} to both factors and grouping the frequencies $(\pm r,\pm\ell)$ gives the coefficients
\begin{equation}
  \alpha_{\ell r}=
  \begin{cases}
    4i^r
    J_{(\ell+r)/2}(-\pi/4)
    J_{(r-\ell)/2}(-\pi/4),
    & \ell\equiv r\pmod 2,\\
    0,&\text{otherwise}.
  \end{cases}
  \label{eq:bessel-coeff}
\end{equation}
In the notation of \cite{AKY26}, this is the window parameter $\gamma=1/2$.  Chebyshev series conventionally halve their $n=0$ term.  We absorb this weight into the coefficients by setting
\begin{equation}
  \eta_n=
  \begin{cases}
    \tfrac12,&n=0,\\
    1,&n\ge1,
  \end{cases}
  \qquad
  \alpha'_{\ell r}=\eta_\ell\eta_r\alpha_{\ell r}.
  \label{eq:endpoint-adjusted-coeff}
\end{equation}
Then the expansion reads
\begin{equation}
  H(x,z)=\sum_{r=0}^{\infty}\sum_{\ell=0}^{\infty}
  \alpha'_{\ell r}T_\ell(z)T_r(x),
  \label{eq:exact-bivariate-expansion}
\end{equation}
and holds uniformly on $[-1,1]^2$. For more details refer to  \cite[Lemma~A.3]{Tow14} and \cite[Appendix~A]{RT18}.  The coefficients $\alpha'_{\ell r}$ coincide with those of \cite{AKY26}.  By \Cref{lm:bessel_bound}, $|J_n(\pi/4)|\le e^{\pi^2/64}(\pi/8)^n/n!$, so the coefficients decay faster than any geometric sequence and a small truncation rank $K$ suffices.  For $K\ge1$ let
\begin{equation}
  H_K(x,z)=\sum_{r=0}^{K-1}\sum_{\ell=0}^{K-1}\alpha'_{\ell r}T_\ell(z)T_r(x)
  \label{eq:rect-truncation}
\end{equation}
be the rectangular truncation of \Cref{eq:exact-bivariate-expansion}.  \cite{RT18} show that $\max_{x,z\in[-1,1]}|H(x,z)-H_K(x,z)|\le\eps$ for a rank $K=O\bigl(\log(1/\eps)/\log\log(1/\eps)\bigr)$ given by a Lambert-$W$ expression.  Hence the matrix with entries $e^{-i\pi(k/N)y_j}$ is approximated to within $\eps$ in every entry by the rank-$K$ matrix with entries $e^{-i\pi y_j/2}H_K(w_k,y_j)$.

\begin{remark}
\label{rem:rt18-changes}
Two features of this result matter for us.
\begin{itemize}
  \item \emph{Entrywise bound.}  The bound is entrywise, while a block encoding needs the operator norm.  Since $F_t$ carries the normalization $1/\sqrt N$, a kernel error $\Delta_K$ in every entry gives
  \begin{equation}
    \norm{F_t-F_t^{(K)}}\le\norm{F_t-F_t^{(K)}}_{\rm F}\le\sqrt N\,\Delta_K
    \label{eq:entrywise-to-operator}
  \end{equation}
  (Bound~1 in the proof of \Cref{thm:dyadic-nuqft}).  The result of \cite{RT18} must therefore be applied with $\eps/\sqrt N$ in place of $\eps$, so $K$ depends on $N$ through $\log(\sqrt N/\eps)$.
  \item \emph{Incomplete tail bound.}  The proof in \cite{RT18} bounds the truncation error by $\sum_{p\ge K}\sum_{r\ge K}|a_{pr}|$, the sum over pairs in which \emph{both} indices are at least $K$.  The rectangular truncation also discards the pairs in which exactly one index is at least $K$, and these are not covered.  \Cref{lem:bivariate-chebyshev-bessel} bounds the full discarded sum, over $\max\{\ell,r\}\ge K$, and replaces their Lambert-$W$ rank by an explicit $K=O(\log(N/\eps))$.  The gain is a complete bound with explicit constants, not a smaller rank: with $\eps/\sqrt N$ substituted, their rank is $O\bigl(\log(N/\eps)/\log\log(N/\eps)\bigr)$.
\end{itemize}
\end{remark}

\paragraph{Assemble the low-rank factorization.}
Truncating \Cref{eq:exact-bivariate-expansion} to $0\le\ell,r<K$ and combining with \Cref{eq:background-phase-split,eq:residual-kernel} gives
\begin{equation}
  e^{-2\pi ikt_j}
  \approx
  e^{-2\pi ik\sigma_j/N}\sum_{r=0}^{K-1}u_r(j)\,v_r(k),
  \qquad
  u_r(j)=e^{-i\pi y_j/2}\sum_{\ell=0}^{K-1}\alpha'_{\ell r}T_\ell(y_j),
  \quad
  v_r(k)=T_r(w_k).
  \label{eq:background-low-rank}
\end{equation}
In matrix form, with rows indexed by frequencies and columns by nodes as in \Cref{eq:type-II-mat}, so that $(F_t)_{kj}=\frac{1}{\sqrt N}e^{-2\pi ikt_j}$, let $D_{u_r}$ and $D_{v_r}$ be the diagonal matrices with entries $u_r(j)$ and $v_r(k)$, let $F$ be the uniform DFT, and let $M_\sigma$ be as in \Cref{eq:prelim-routing-matrix-data}, so that $FM_\sigma$ places column $\sigma_j$ of $F$ in column $j$.  Then
\begin{equation}
  F_t\approx\sum_{r=0}^{K-1}D_{v_r}FM_\sigma D_{u_r}.
  \label{eq:background-matrix-low-rank}
\end{equation}
Each term is a node scaling, a column selection, one uniform DFT and a frequency scaling.  Classically this costs $O(KN\log N)$ in total, and several nodes sharing a grid point cost nothing extra.  \Cref{sec:aky26-circuit} summarizes how \cite{AKY26} implements each of these factors using a quantum circuit.

\subsection{The NUQFT circuit and the parameter \texorpdfstring{$\kappa$}{kappa}}
\label{sec:aky26-circuit}
\cite{AKY26} turns the matrix factorization \Cref{eq:background-matrix-low-rank} into a block encoding.  Each term $D_{v_r}FM_\sigma D_{u_r}$ is a product of block encodings of a frequency diagonal, the QFT, the matrix $M_\sigma$, which sends node $j$ to grid point $\sigma_j$, and a node diagonal. These $K$ terms are combined by a linear combination of unitaries (LCU) \cite{CW12}.  The diagonal entries are built from Chebyshev values $T_n(y)=\cos(n\arccos y)$, obtained by computing $\arccos$ into a register and applying controlled rotations.  We now explain where $\kappa$ enters the resulting complexity.

The normalization of such an LCU is governed by the coefficient magnitudes.  For a truncation rank $K$, write
\begin{equation}
  \lambda_r=\sum_{\ell=0}^{K-1}|\alpha'_{\ell r}|,
  \qquad
  \Lambda=\sum_{r=0}^{K-1}\lambda_r=\sum_{r,\ell<K}|\alpha'_{\ell r}|,
  \label{eq:lambda-r}
\end{equation}
so that $\lambda_r$ is the weight of the $r$th term and $\Lambda$ the total weight.  We now explain where $\kappa$ enters the analysis.

In the analysis of \cite{AKY26}, finite-precision access produces an $m$-bit approximation
$\widehat y_j$ satisfying 
\begin{equation}
  |y_j-\widehat y_j|\le N2^{-m+1}.
  \label{eq:source-residual-perturbation}
\end{equation}
Writing $\theta_j=\arccos(y_j)$ and letting $\widehat\theta_j$ be its
$p$-bit approximation, the mean-value theorem is applied to the
input perturbation of $\arccos$.  Thus, for some point $y_j^*$ between
$y_j$ and $\widehat y_j$,
\begin{align}
  |\theta_j-\widehat\theta_j|
  &\le
  |\arccos(y_j)-\arccos(\widehat y_j)|
  +|\arccos(\widehat y_j)-\widehat\theta_j| \notag\\
  &\le
  \frac{N2^{-m+1}}{\sqrt{1-(y_j^*)^2}}
  +O(2^{-p}).
  \label{eq:source-acos-input-perturbation}
\end{align}
In the second inequality, the first term is the mean-value theorem with $|\tfrac{d}{dy}\arccos y|=1/\sqrt{1-y^2}$, combined with \Cref{eq:source-residual-perturbation}. The second is the output error of the reversible $\arccos$ circuit, part~(i) of \Cref{thm:reversible-elementary-functions} \cite{BBG24}.  Both appear in the proof of \cite[Lemma~27]{AKY26}.

Consequently, a degree-$\ell$ Chebyshev polynomial of first kind incurs error
\begin{equation}
  |T_\ell(y_j)-T_\ell(\widehat y_j)|
  = | \cos(\ell \theta_j) - \cos(\ell \widehat \theta_j)|
  \leq^{(a)} \ell |\theta_j - \widehat \theta_j|
  \le
  \ell\left(
    O(2^{-p})+
    \frac{N2^{-m+1}}{\sqrt{1-(y_j^*)^2}}
  \right),
  \label{eq:source-chebyshev-kappa-error}
\end{equation}
where step $(a)$ follows from an application of the mean-value theorem to the cosine function. Their result therefore introduces
\begin{equation}
  \kappa
  :=\max_j\frac{1}{\sqrt{1-(y_j^*)^2}},
  \label{eq:kappa}
\end{equation}
and its sufficient precision and gate bounds contain logarithmic factors of
the form $\log(1+K\kappa)$.

The point $y_j^*$ in \Cref{eq:kappa} comes from the mean-value theorem in the proof of \cite[Lemma~27]{AKY26}.  It exists, but its position between $y_j$ and $\widehat y_j$ is not specified, so $\kappa$ formally depends on $m$ as well.  By \Cref{eq:source-residual-perturbation}, $\kappa$ can be bounded independently of $m$ whenever $\max_j|y_j|<1$ and $m$ is large enough.  \cite{AKY26} also report numerical experiments (their Table~1, $N=2^6$, $\eps=10^{-12}$) in which the minimal required precision $m^*$ stays essentially unchanged as $\kappa$ grows, which suggests that the dependence of $m$ on $\kappa$ is mild in practice.

With this parameter, the NUQFT block-encodes the Type-II NUDFT to accuracy $\eps$ using
\begin{equation}
  O\bigl(n^2+L^2+\log(1+\kappa L)\bigr)\ \text{gates}
  \quad\text{and}\quad
  O\bigl(L+\log(1+\kappa L)\bigr)\ \text{qubits},
  \qquad L:=n+\log(1/\eps).
  \label{eq:nuqft-kappa-cost}
\end{equation}
This dependence is undesirable in a general complexity theorem, even though
it is only logarithmic.  First, $\kappa$ measures the conditioning of the
auxiliary map $y\mapsto\arccos y$, rather than the spectral conditioning of
the NUDFT matrix itself.  Second, $\arccos'(y)=-(1-y^2)^{-1/2}$ is singular
at $y=\pm1$, so admissible nodes can make $\kappa$ arbitrarily large by
approaching a nearest-grid cell boundary.  A polylogarithmic bound in
$n=\log_2N$ therefore follows from the estimate only after one
controls the growth of $\kappa$.  Third, certifying such control is
instance-dependent: it amounts to lower-bounding the distance of every
scaled node $Nt_j$ from a half-integer.  

\paragraph{Issues addressed in this paper.}
\label{sec:background-issues}
Besides $\kappa$, we note one further point in \cite{AKY26}.  Both are resolved in \Cref{sec:CFNUQFT}.
\begin{enumerate}
  \item \emph{Nodes near $t=1$.} \cite[Section~3.1.2]{AKY26} forms the offset from the wrapped index, $2N(t_j-\sigma_j/N)$, instead of from $s_j^{\mathrm{unw}}$, so for nodes within $1/(2N)$ of $1$ this quantity is close to $2N$, outside the domain of $\arccos$.  We form it from the unwrapped index, as in \Cref{eq:dyadic-routing-data-a}, and use the wrapped index only to select the column of the DFT.
  \item \emph{Conditioning and normalization.} The only $\kappa$-dependent step is the mean-value bound \Cref{eq:source-acos-input-perturbation} on the node diagonal \cite[Lemma~27]{AKY26}.  \Cref{lem:diagonal-factor-implementation} replaces it by a bound that is uniform in the nodes, and the weighted outer LCU of \Cref{lem:assembled-factor-error} improves the block-encoding normalization from $O(K^2\sqrt{d_r})$ to $O(\sqrt{d_r})$, where $d_r$ bounds the number of nodes that round to the same grid point.
\end{enumerate}

\section{Conditioning-Free NUQFT}\label{sec:CFNUQFT}

We first summarize the strategy for removing the $\kappa$ dependence and then prove the modified lemmas in \Cref{sec:nuqft-preliminary-lemmas}, finally we state and prove the main result in \Cref{sec:exact-dyadic-nuqft}.

\paragraph{The conditioning-free strategy.}
The real nodes $t=(t_0,\ldots,t_{N-1})$ may be irrational.  The NUCT nodes
$t_j=\frac1{2\pi}\arccos(-1+2j/N)$ are an example.  A quantum circuit cannot
store such numbers exactly.  What it does store is an $m$-bit binary value
$\tau_j=a_j/2^m$ for each node.  Our key idea is to treat these stored values
as the \emph{actual} input of the transform, rather than as noisy copies of
the real nodes.  The NUQFT subroutine therefore implements the matrix
$F_\tau$ built from the stored nodes, while the real-node matrix $F_t$
remains the final goal.  This splits the total error into two independent
parts:
\begin{equation}
  \norm{\matU_{\mathrm{II}}-F_t}
  \le
  \norm{\matU_{\mathrm{II}}-F_\tau}
  +
  \norm{F_\tau-F_t}.
  \label{eq:conditioning-free-two-level-decomposition}
\end{equation}
Concretely, the construction keeps the circuit of \cite{AKY26} and changes
the analysis in three steps.
\begin{enumerate}[label=(\roman*),leftmargin=2.4em]
  \item \emph{Split the stored nodes exactly.}  Because $N=2^q$ is a power
  of two, the split \Cref{eq:dyadic-routing-data-a} of a stored node $\tau_j$
  into $s_j^{\mathrm{unw}}$, $\sigma_j$ and $z_j$ is exact bit arithmetic,
  so the phase splitting \Cref{eq:background-phase-split} is an exact 
  identity for $F_\tau$.

  \item \emph{Round only the output angle.}  The offset $z_j=(a_j-2^{m-q}s_j^{\mathrm{unw}})/2^{m-q-1}$ is
  a signed fixed-point number with $m-q-1$ fractional bits, computed
  exactly from $\tau_j$, so the circuit computes $\arccos$ of an
  exact input.  The reversible arithmetic of
  \Cref{thm:reversible-elementary-functions} returns $\widetilde\phi_j$ with
  $|\widetilde\phi_j-\arccos z_j|\le c_{\rm acos}2^{-p}$, and therefore, for
  $0\le\ell<K$,
  \begin{equation}
    |T_\ell(z_j)-\cos(\ell\widetilde\phi_j)|
    \le \ell c_{\rm acos}2^{-p}
    \le Kc_{\rm acos}2^{-p}.
    \label{eq:overview-output-angle-error}
  \end{equation}
  This is where the analysis differs from \cite{AKY26}.  There, the target
  matrix is $F_t$, whose Chebyshev values are $T_\ell(y_j)$ at the true
  offsets $y_j$, but the circuit evaluates $\arccos$ at a rounded offset
  $\widehat y_j$.  Comparing the two passes the input error
  $|y_j-\widehat y_j|$ through $\arccos$ where the slope
  $|\tfrac{d}{dy}\arccos y|=1/\sqrt{1-y^2}$ is unbounded at $y=\pm1$.
  This is the source of $\kappa$ in \Cref{eq:source-acos-input-perturbation}.
  Here, the target of the first part of
  \Cref{eq:conditioning-free-two-level-decomposition} is $F_\tau$, whose
  Chebyshev values are $T_\ell(z_j)$ at the stored offsets.  The circuit
  feeds exactly this $z_j$ into $\arccos$.  The target and the circuit
  therefore use the same input, and the only error is the rounding of the
  output angle.  \Cref{eq:overview-output-angle-error} bounds it by
  $Kc_{\rm acos}2^{-p}$ for every $z_j\in[-1,1]$, including $z_j$ near
  $\pm1$, so $\kappa$ does not appear.  The difference between $t_j$ and
  $\tau_j$ is not lost.  It is moved to the second part of
  \Cref{eq:conditioning-free-two-level-decomposition}, where step~(iii)
  bounds it through the Fourier phases instead of through $\arccos$.

  \item \emph{Count the node error separately.}  Moving a node by a small
  amount changes each Fourier phase by a proportional amount, which gives
  the second part of \Cref{eq:conditioning-free-two-level-decomposition}:
  \begin{equation}
    \|F_t-F_\tau\|
    \le
    2\pi\sqrt{\sum_{k=0}^{N-1}k^2}\,
    \|t-\tau\|_\infty
    \le
    \frac{2\pi}{\sqrt3}N^{3/2}\|t-\tau\|_\infty.
    \label{eq:overview-external-node-perturbation}
  \end{equation}
  This bound, proved within the proof of \Cref{thm:dyadic-nuqft}, uses only
  the Fourier phases and not $\arccos$.  Thus $m=\frac32q+O(\log(1/\eps))$
  bits per node suffice, with no condition on where the nodes sit.
\end{enumerate}
In short, the circuit is essentially that of \cite{AKY26} with a modified error analysis.

\paragraph{Overview of the section.}  \Cref{sec:nuqft-preliminary-lemmas} splits the
error into truncation, implementation and node terms. We prove four lemmas that bound the error of each factor: the coefficient bounds, the nearest-point matrix, the
$\kappa$-free diagonal bound, and the assembly of the $K$ terms. \Cref{sec:nuqft-oracles} lists the oracles the circuit needs and shows that the coefficient preparation does not depend on the nodes, so it can be built once and reused.  \Cref{sec:exact-dyadic-nuqft} states and proves the main result, \Cref{thm:dyadic-nuqft}: an $\eps$-accurate block encoding of $F_t$ with normalization $O(\sqrt{d_r})$ and $\widetilde O(m^2+L^2)$ gates. \Cref{subsec:uniform-versus-weighted-outer-lcu} compares the two ways of assembling the $K$ terms.

\subsection{Error analysis of block encoding factors}
\label{sec:nuqft-preliminary-lemmas}

Let $\tau_j$ be the stored nodes, split as in \Cref{eq:dyadic-routing-data-a} into $\sigma_j$ and $z_j$.  The circuit targets
\begin{equation}
  (F_\tau)_{kj}:=\frac{1}{\sqrt N}e^{-2\pi i k\tau_j}
  \label{eq:proof-computational-type-II}
\end{equation}
through the rank-$K$ factorization \Cref{eq:background-matrix-low-rank}, with $z_j$ in place of $y_j$:
\begin{equation}
  v_r(k):=T_r(w_k),
  \quad
  u_r(j):=e^{-i\pi z_j/2}\sum_{\ell=0}^{K-1}\alpha'_{\ell r}T_\ell(z_j),
  \quad
  F_\tau^{(K)}:=\sum_{r=0}^{K-1}D_{v_r}FM_\sigma D_{u_r},
  \label{eq:dyadic-direct-factorization}
\end{equation}
where $w_k=2k/N-1$, $D_{v_r}=\diag(v_r(0),\ldots,v_r(N-1))$ and $D_{u_r}=\diag(u_r(0),\ldots,u_r(N-1))$.  The circuit $\matU_{\mathrm{II}}$ of \Cref{thm:dyadic-nuqft} is built from the factorization $F_\tau^{(K)}$, but every Chebyshev value in it is computed from a finite-precision angle.  The matrix it encodes is therefore $\sum_r\widetilde D_{v_r}FM_\sigma\widetilde D_{u_r}$, where $\widetilde D_{v_r}$ and $\widetilde D_{u_r}$ are $D_{v_r}$ and $D_{u_r}$ with each $T_n(w)$ replaced by $\widetilde T_n(w)=\cos(n\widetilde\phi(w))$. This matrix differs slightly from $F_\tau^{(K)}$. By the triangle inequality, its error relative to the target $F_t$ splits into three terms,
\begin{equation}
  \norm{\matU_{\mathrm{II}}-F_t}
  \le
  \norm{F_\tau-F_\tau^{(K)}}
  +\norm{\matU_{\mathrm{II}}-F_\tau^{(K)}}
  +\norm{F_\tau-F_t},
  \label{eq:error-budget}
\end{equation}
which refines \Cref{eq:conditioning-free-two-level-decomposition}.  The first is the truncation error of the Chebyshev-Bessel expansion and depends only on $K$.  The second is the implementation error, which comes from computing $\arccos$ to $p$ bits.  The third is the node error from storing $t_j$ with $m$ bits.  It is bounded directly in the proof of \Cref{thm:dyadic-nuqft}, through the Fourier phases alone.  The normalization of $\matU_{\mathrm{II}}$ is the product of the normalizations of the factors in each term, combined by an LCU over $r$.

The four lemmas below supply these pieces.  \Cref{lem:bivariate-chebyshev-bessel} bounds the expansion coefficients: its tail bound controls the truncation term, and its bound on the total weight $\Lambda$ controls the normalization.  \Cref{lem:routing-factorization} gives the normalization $\sqrt{d_r}$ contributed by $M_\sigma$ and shows that the offsets $z_j$ are computed exactly.  \Cref{lem:diagonal-factor-implementation} uses this exactness to bound the error of the Chebyshev diagonals without $\kappa$.  \Cref{lem:assembled-factor-error} combines them into a bound on the implementation term and the final normalization.

\paragraph{Coefficients.}
The coefficients $\alpha'_{\ell r}$ of \Cref{eq:bessel-coeff,eq:endpoint-adjusted-coeff} enter twice: the tail of the expansion gives the truncation error, and the total weight $\Lambda$ of \Cref{eq:lambda-r} sets the normalization of the LCU.  Both follow from one estimate.  Part~(i) replaces the Lambert-$W$ truncation rank of \cite{RT18} with an explicit $K$ and bounds the full truncation error (\Cref{rem:rt18-changes}), and part~(ii) improves the bound $\Lambda=O(K)$ of \cite{AKY26} to a constant.

\begin{lemma}[Chebyshev-Bessel coefficients]
\label{lem:bivariate-chebyshev-bessel}\label{lem:constant-bessel-mass}
Let $a=\pi/4$ and $c_\alpha:=8e^{a^2/2+a}=8\exp(\pi^2/32+\pi/4)<24$, and let $H_K$ be the rectangular truncation \Cref{eq:rect-truncation}.
\begin{enumerate}[label=(\roman*)]
  \item \emph{Truncation.}
  \begin{equation}
    \Delta_K:=\max_{x,z\in[-1,1]}|H(x,z)-H_K(x,z)|
    \le c_\alpha\frac{a^K}{K!}.
    \label{eq:bivariate-tail-bound}
  \end{equation}
  In particular, for $N\ge2$ and $0<\delta\le1$, any $K\ge\lceil\log_2(24\sqrt N/\delta)\rceil$ gives $\sqrt N\,\Delta_K\le\delta$.
  \item \emph{Total weight.}  For every $K$,
  \begin{equation}
    \Lambda\le c_\alpha<24.
    \label{eq:constant-bessel-mass}
  \end{equation}
\end{enumerate}
\end{lemma}

\begin{proof}
Both parts bound sums of $|\alpha'_{\ell r}|$ over sets of index pairs.  The coefficients vanish unless $\ell\equiv r\pmod2$.  For the others, $\eta_n\le1$ and the symmetries \Cref{eq:bessel-symmetries} give
\begin{equation}
  |\alpha'_{\ell r}|\le4|J_u(a)|\,|J_v(a)|,
  \qquad
  u=\tfrac{\ell+r}{2},\quad v=\tfrac{|\ell-r|}{2},
  \label{eq:alpha-bessel-bound}
\end{equation}
with $u,v$ nonnegative integers and $u+v=\max\{\ell,r\}$.  The map $(\ell,r)\mapsto(u,v)$ has at most two preimages, $(u+v,u-v)$ and $(u-v,u+v)$.  By \Cref{lm:bessel_bound}, $|J_u(a)J_v(a)|\le e^{a^2/2}(a/2)^{u+v}/(u!\,v!)$, and summing over $u+v=s$ with $\sum_u\binom su=2^s$ gives $\sum_{u+v=s}|J_u(a)J_v(a)|\le e^{a^2/2}a^s/s!$.  Hence, for every $s_0\ge0$,
\begin{equation}
  \sum_{\max\{\ell,r\}\ge s_0}|\alpha'_{\ell r}|
  \le8e^{a^2/2}\sum_{s\ge s_0}\frac{a^s}{s!}.
  \label{eq:bessel-shell-sum}
\end{equation}

(ii) Taking $s_0=0$ gives $\Lambda\le8e^{a^2/2}e^a=c_\alpha$, and $c_\alpha\approx23.89$.

(i) Since $|T_n|\le1$ on $[-1,1]$ and the expansion converges absolutely, $|H-H_K|$ is at most the sum of $|\alpha'_{\ell r}|$ over the pairs outside the square $0\le\ell,r<K$, that is, over $\max\{\ell,r\}\ge K$.  By \Cref{eq:bessel-shell-sum} with $s_0=K$ and the Taylor remainder bound $\sum_{s\ge K}a^s/s!\le e^aa^K/K!$, this gives \Cref{eq:bivariate-tail-bound}.  For $K\ge5>2ea\approx4.27$, the estimate $K!\ge(K/e)^K$ gives $a^K/K!\le(ea/K)^K\le2^{-K}$, so $\Delta_K<24\cdot2^{-K}$.  For $N\ge2$ and $\delta\le1$ the stated $K$ is at least $\lceil\log_2(24\sqrt2)\rceil=6$, and then $\sqrt N\,\Delta_K<24\sqrt N\,2^{-K}\le\delta$.
\end{proof}

Part~(i) is used in Bound~1 of the proof of \Cref{thm:dyadic-nuqft}, and part~(ii) in the normalization and the precision choice of \Cref{lem:assembled-factor-error}.

\paragraph{The nearest-point matrix.}
The factor $M_\sigma$ is the only one with normalization depending on how the nodes are distributed, i.e. the normalization grows with the number of nodes that share a grid point.  The next lemma records this normalization, and the cost of computing the split \Cref{eq:dyadic-routing-data-a} which is exact for stored nodes.

\begin{lemma}[Nearest-point matrix and sparse normalization]
\label{lem:routing-factorization}
Let $\tau_j\in[0,1)$ be exact nodes, with $s_j^{\mathrm{unw}}$, $\sigma_j$, $z_j$, $M_\sigma$ and $d_\tau$ as in \Cref{eq:dyadic-routing-data-a,eq:prelim-routing-matrix-data}.  Then
\begin{equation}
  \norm{M_\sigma}=\sqrt{d_\tau}.
  \label{eq:routing-matrix-norm}
\end{equation}
Given sparse access to $M_\sigma$ with a known bound $d_r\ge d_\tau$ on its row sparsity, \cite[Lemma~48]{GSLW18} gives an exact block encoding of $M_\sigma$ with normalization $\sqrt{d_r}$, and its adjoint block-encodes $M_\sigma^{\mathsf T}$ with the same normalization.

If $N=2^q$ and $\tau_j=a_j/2^m$ is an $m$-bit fixed-point number, then $s_j^{\mathrm{unw}}$, $\sigma_j$, and $z_j$ are computed exactly with $O(m+q)$ qubits and $O(m+q)$ Toffoli and CNOT gates. In particular, $z_j$ carries no approximation error.
\end{lemma}

\begin{proof}
For the norm, $M_\sigma M_\sigma^\dagger=\sum_j\ket{\sigma_j}\!\bra{\sigma_j}=\sum_\ell d_\ell\ket\ell\!\bra\ell$, so $\norm{M_\sigma}^2=\max_\ell d_\ell=d_\tau$. The matrix $M_\sigma$ has column sparsity one, row sparsity at most $d_r$, and entries in $\{0,1\}$, so the sparse-access block encoding \cite[Lemma~48]{GSLW18} has normalization $\sqrt{s_rs_c}=\sqrt{d_r\cdot1}$ and zero error.  Taking the adjoint circuit gives the statement for $M_\sigma^{\mathsf T}=M_\sigma^\dagger$, whose sparsities are swapped, leaving $s_rs_c$ unchanged.

For fixed-point nodes, multiplying by $N=2^q$ is a binary shift. The shifted number determines $s_j^{\mathrm{unw}}$ by an exact comparison (with a fixed tie convention), its low $q$ bits give $\sigma_j$, and an exact signed subtraction gives $z_j$. These are reversible ripple-carry operations of width $O(m+q)$, and reversing them uncomputes the work registers.
\end{proof}

The exactness of $z_j$ is what the next lemma relies on.

\paragraph{The Chebyshev diagonals.}
The diagonals contain Chebyshev values $T_n(w)=\cos(n\arccos w)$, at $w=w_k$ in $D_{v_r}$ and at $w=z_j$ in $D_{u_r}$.  The circuit computes an angle $\widetilde\phi(w)\approx\arccos w$ into a register and applies a controlled rotation, so it implements
\begin{equation}
  \widetilde T_n(w):=\cos\bigl(n\widetilde\phi(w)\bigr)
  \label{eq:approx-chebyshev-definition}
\end{equation}
and the diagonals
\[
  \widetilde D_{v_r}:=\diag\bigl(\widetilde T_r(w_k)\bigr)_{k},
  \qquad
  \widetilde D_{u_r}:=\diag\Bigl(e^{-i\pi z_j/2}\sum_{\ell=0}^{K-1}\alpha'_{\ell r}\widetilde T_\ell(z_j)\Bigr)_{j}.
\]
This is the step where \cite{AKY26} incurs $\kappa$. Here both inputs of $\arccos$ are exact: $w_k=(k-N/2)/2^{q-1}$ is a signed fixed-point number with $q-1$ fractional bits, and $z_j$ is computed exactly from $\tau_j$ by \Cref{lem:routing-factorization}.  The only error left is in the output angle.  By part~(i) of \Cref{thm:reversible-elementary-functions}, running the $\arccos$ circuit with $\lceil\log_2c_{\rm acos}\rceil$ additional bits makes this error at most $2^{-p}$.

\begin{lemma}[Frequency and node diagonal errors]
\label{lem:diagonal-factor-implementation}
Suppose that, on every exact input $w\in[-1,1]$, the angle circuit returns $\widetilde\phi(w)$ with
\begin{equation}
  \left|\widetilde\phi(w)-\arccos w\right|\le\delta_\phi.
  \label{eq:angle-output-error-lemma37}
\end{equation}
Then, for $0\le n,r<K$:
\begin{enumerate}[label=(\roman*)]
  \item $|\widetilde T_n(w)-T_n(w)|\le n\delta_\phi\le K\delta_\phi$.
  \label{item:cheb-from-angle}
  \item The frequency and node diagonals satisfy
  \begin{equation}
    \norm{\widetilde D_{v_r}-D_{v_r}}\le K\delta_\phi,
    \label{eq:frequency-diagonal-error}
  \end{equation}
  \begin{equation}
    \norm{\widetilde D_{u_r}-D_{u_r}}\le\lambda_rK\delta_\phi.
    \label{eq:diagonal-factor-errors}
  \end{equation}
  \item $\widetilde D_{v_r}$ and $\widetilde D_{u_r}$ have block encodings with normalizations $1$ and $\lambda_r$, respectively.
\end{enumerate}
\end{lemma}

\begin{proof}
(i) Since $T_n(w)=\cos(n\arccos w)$ and the cosine is $1$-Lipschitz, $|\cos(n\widetilde\phi)-\cos(n\arccos w)|\le n|\widetilde\phi-\arccos w|\le n\delta_\phi$.\\

(ii) The matrices are diagonal, so their norms are the largest entry errors.  For $D_{v_r}$ this is (i) with $n=r$.  For $D_{u_r}$, the phase $e^{-i\pi z_j/2}$ has unit modulus, so by (i)
\[
  \Bigl|\sum_{\ell<K}\alpha'_{\ell r}\bigl(\widetilde T_\ell(z_j)-T_\ell(z_j)\bigr)\Bigr|
  \le\sum_{\ell<K}|\alpha'_{\ell r}|\,K\delta_\phi
  =\lambda_rK\delta_\phi.
\] \\

(iii) We use a data register $\ket{\cdot}$, a work register $\ket{\cdot}_{\rm w}$ that holds the offset and the angle, a coefficient register $\ket{\cdot}_{\rm c}$ of $\lceil\log_2K\rceil$ qubits, and one ancilla qubit $\ket{\cdot}_{\rm a}$.  With $R_y(\vartheta)=e^{-i\vartheta Y/2}$,
\begin{equation}
  R_y(2\vartheta)\ket0_{\rm a}=\cos\vartheta\,\ket0_{\rm a}+\sin\vartheta\,\ket1_{\rm a}.
  \label{eq:ry-action}
\end{equation}
For an integer $n$ and a fixed-point angle $\phi$, the product $n\phi$ is computed exactly, so a rotation $R_y(2n\phi)$ controlled on the angle register introduces no new error.

Compute $w_k$ from $k$ and its angle by part~(i) of \Cref{thm:reversible-elementary-functions}, rotate the ancilla by \Cref{eq:ry-action} with $\vartheta=r\widetilde\phi(w_k)$, and uncompute the work register:
\begin{align*}
  \ket k\ket0_{\rm w}\ket0_{\rm a}
  &\mapsto\ket k\ket{w_k,\widetilde\phi(w_k)}_{\rm w}\ket0_{\rm a}\\
  &\mapsto\ket k\ket{w_k,\widetilde\phi(w_k)}_{\rm w}\bigl(\widetilde T_r(w_k)\ket0_{\rm a}+\sin(r\widetilde\phi(w_k))\ket1_{\rm a}\bigr)\\
  &\mapsto\widetilde T_r(w_k)\ket k\ket0_{\rm w}\ket0_{\rm a}+\sin(r\widetilde\phi(w_k))\ket k\ket0_{\rm w}\ket1_{\rm a}.
\end{align*}
Projecting the ancilla onto $\ket0_{\rm a}$ leaves $\widetilde T_r(w_k)\ket k$, which is a block encoding of $\widetilde D_{v_r}$ with norm $1$.

Now for the block encoding of $\widetilde{D}_{u_r}$, write $\alpha'_{\ell r}=|\alpha'_{\ell r}|e^{i\varphi_{\ell r}}$ and use an LCU over $\ell$.  Let $\PREP_r$ be a unitary on the coefficient register with
\begin{equation}
  \PREP_r\ket0_{\rm c}=\frac{1}{\sqrt{\lambda_r}}\sum_{\ell=0}^{K-1}\sqrt{|\alpha'_{\ell r}|}\,\ket\ell_{\rm c}.
  \label{eq:lemma-prep-r}
\end{equation}
\Cref{sec:nuqft-oracles} lists $\PREP_r$ among the oracles the NUQFT uses, and \Cref{prop:coefficient-prep} gives a circuit for it with $O(K)$ gates.  Write $\widetilde\phi_j:=\widetilde\phi(z_j)$.  The circuit acts as follows.
\begin{enumerate}[label=(\alph*),leftmargin=2.4em]
  \item \emph{Offset and angle.}  From $\tau_j$, compute $z_j$ exactly (\Cref{lem:routing-factorization}) and then $\widetilde\phi_j$ (part~(i) of \Cref{thm:reversible-elementary-functions}):
  \[
    \ket j\ket0_{\rm w}\mapsto\ket j\ket{z_j,\widetilde\phi_j}_{\rm w}.
  \]
  \item \emph{Prepare.}  Apply $\PREP_r$ to the coefficient register, giving $\lambda_r^{-1/2}\sum_\ell\sqrt{|\alpha'_{\ell r}|}\ket j\ket{z_j,\widetilde\phi_j}_{\rm w}\ket\ell_{\rm c}\ket0_{\rm a}$.
  \item \emph{Select.}  Controlled on $\ket\ell_{\rm c}$ and on the angle in the work register, apply the phase $e^{i\varphi_{\ell r}}$ and the rotation $R_y(2\ell\widetilde\phi_j)$ to the ancilla.  By \Cref{eq:ry-action},
  \[
    \ket{z_j,\widetilde\phi_j}_{\rm w}\ket\ell_{\rm c}\ket0_{\rm a}
    \mapsto
    e^{i\varphi_{\ell r}}\ket{z_j,\widetilde\phi_j}_{\rm w}\ket\ell_{\rm c}
    \bigl(\widetilde T_\ell(z_j)\ket0_{\rm a}+\sin(\ell\widetilde\phi_j)\ket1_{\rm a}\bigr).
  \]
  \item \emph{Unprepare.}  Apply $\PREP_r^\dagger$.  Since the amplitudes in \Cref{eq:lemma-prep-r} are real and nonnegative, $\bra0_{\rm c}\PREP_r^\dagger\ket\ell_{\rm c}=\sqrt{|\alpha'_{\ell r}|/\lambda_r}$.  The component with $\ket0_{\rm c}\ket0_{\rm a}$ is therefore
  \[
    \sum_{\ell=0}^{K-1}\frac{|\alpha'_{\ell r}|}{\lambda_r}e^{i\varphi_{\ell r}}\widetilde T_\ell(z_j)\,\ket j\ket{z_j,\widetilde\phi_j}_{\rm w}
    =\frac{1}{\lambda_r}\sum_{\ell=0}^{K-1}\alpha'_{\ell r}\widetilde T_\ell(z_j)\,\ket j\ket{z_j,\widetilde\phi_j}_{\rm w}.
  \]
  \item \emph{Node phase.}  Since $z_j$ is a signed fixed-point number, $z_j=\sum_ic_ib_{j,i}$ with bits $b_{j,i}\in\{0,1\}$ and fixed signed powers of two $c_i$.  So $e^{-i\pi z_j/2}=\prod_ie^{-i\pi c_ib_{j,i}/2}$ is applied by one phase gate on each bit of $z_j$:
  \[
    \ket{z_j,\widetilde\phi_j}_{\rm w}\mapsto e^{-i\pi z_j/2}\ket{z_j,\widetilde\phi_j}_{\rm w}.
  \]
  \item \emph{Uncompute.}  Reverse step~(a), which returns the work register to $\ket0_{\rm w}$ in every branch, since it was used only as a control.
\end{enumerate}
Altogether, projecting the coefficient register and the ancilla onto $\ket0_{\rm c}\ket0_{\rm a}$ maps $\ket j\ket0_{\rm w}$ to
\[
  \frac{1}{\lambda_r}\,e^{-i\pi z_j/2}\sum_{\ell=0}^{K-1}\alpha'_{\ell r}\widetilde T_\ell(z_j)\,\ket j\ket0_{\rm w}
  =\frac{1}{\lambda_r}\bigl(\widetilde D_{u_r}\bigr)_{jj}\ket j\ket0_{\rm w},
\]
so the circuit is a block encoding of $\widetilde D_{u_r}$ with normalization $\lambda_r$.  Here $\lambda_r>0$ for every $r$, since $\alpha'_{rr}=4\eta_r^2i^rJ_r(-\pi/4)J_0(-\pi/4)\neq0$.
\end{proof}

\paragraph{Assembly.}
It remains to combine the factors.  Each term $A_r:=D_{v_r}FM_\sigma D_{u_r}$ is a product of four block encodings (the two diagonals, the QFT and $M_\sigma$), and the $K$ terms are added by an outer LCU with weights proportional to $\lambda_r$.

\begin{lemma}[Error and normalization of the assembled factors]
\label{lem:assembled-factor-error}
Let $\widetilde A_r:=\widetilde D_{v_r}FM_\sigma\widetilde D_{u_r}$, and suppose $K\delta_\phi\le1$.
\begin{enumerate}[label=(\roman*)]
  \item Each term satisfies
  \begin{equation}
    \norm{\widetilde A_r-A_r}\le3\sqrt{d_r}\,\lambda_rK\delta_\phi.
    \label{eq:component-assembly-error}
  \end{equation}
  \item With the outer state $\PREP_{\mathrm{out}}$ of \Cref{eq:paper-universal-prep}, the LCU block-encodes $\sum_r\widetilde A_r$ with normalization $\sqrt{d_r}\Lambda<24\sqrt{d_r}$, and
  \begin{equation}
    \Bigl\|\sum_{r=0}^{K-1}\widetilde A_r-\sum_{r=0}^{K-1}A_r\Bigr\|\le3\sqrt{d_r}\Lambda K\delta_\phi.
    \label{eq:weighted-assembly-error}
  \end{equation}
  \item For $0<\eps\le1$, the choice
  \begin{equation}
    \delta_\phi=2^{-p},
    \qquad
    p:=\left\lceil\log_2\!\left(\frac{144\sqrt{d_r}K}{\eps}\right)\right\rceil
    \label{eq:lemma38-explicit-precision-choice}
  \end{equation}
  satisfies $K\delta_\phi<1$ and makes the error in \Cref{eq:weighted-assembly-error} at most $\eps/2$.
  \label{item:precision-choice}
\end{enumerate}
\end{lemma}

\begin{proof}
(i) This is the triangle-inequality argument of \cite[Section~5.2]{AKY26}, with \Cref{lem:diagonal-factor-implementation}.  Since $F$ is unitary, $\norm{FM_\sigma}=\norm{M_\sigma}\le\sqrt{d_r}$ by \Cref{lem:routing-factorization}.  Then
\[
  \widetilde A_r-A_r=(\widetilde D_{v_r}-D_{v_r})FM_\sigma D_{u_r}+\widetilde D_{v_r}FM_\sigma(\widetilde D_{u_r}-D_{u_r}).
\]
Here $\norm{D_{u_r}}\le\lambda_r$, since $|u_r(j)|\le\sum_\ell|\alpha'_{\ell r}|$, and $\norm{\widetilde D_{v_r}}\le1+K\delta_\phi\le2$.  With \Cref{eq:frequency-diagonal-error,eq:diagonal-factor-errors}, the two terms are at most $K\delta_\phi\sqrt{d_r}\lambda_r$ and $2\sqrt{d_r}\lambda_rK\delta_\phi$.

(ii) The block encodings of $\widetilde D_{v_r}$, $F$, $M_\sigma$ and $\widetilde D_{u_r}$ have normalizations $1$, $1$, $\sqrt{d_r}$ and $\lambda_r$ (\Cref{lem:routing-factorization,lem:diagonal-factor-implementation}), so their product encodes $\widetilde A_r$ with normalization $\beta_r=\sqrt{d_r}\lambda_r$.  The outer state has amplitudes $\sqrt{\lambda_r/\Lambda}$, so the PREP-SELECT-PREP$^\dagger$ sandwich selects
\[
  \sum_{r=0}^{K-1}\frac{\lambda_r}{\Lambda}\cdot\frac{\widetilde A_r}{\sqrt{d_r}\lambda_r}
  =\frac{\sum_r\widetilde A_r}{\sqrt{d_r}\Lambda}.
\]
The normalization is therefore $\sqrt{d_r}\Lambda$, which is less than $24\sqrt{d_r}$ by \Cref{lem:constant-bessel-mass}.  Summing \Cref{eq:component-assembly-error} over $r$ gives \Cref{eq:weighted-assembly-error}.

(iii) The choice gives $K2^{-p}\le\eps/(144\sqrt{d_r})<1$, and with $\Lambda<24$, $3\sqrt{d_r}\Lambda K2^{-p}<72\sqrt{d_r}K2^{-p}\le\eps/2$.
\end{proof}

The weighted outer LCU of part~(ii) gives normalization $O(\sqrt{d_r})$.  A uniform outer LCU, close to the construction of \cite{AKY26}, would give $O(K\sqrt{d_r})$.  \Cref{subsec:uniform-versus-weighted-outer-lcu} compares the two.

\subsection{Oracle assumptions}
\label{sec:nuqft-oracles}

The circuit of \Cref{thm:dyadic-nuqft} is built from three oracles.  The first two describe the node set and must be constructed for each problem.  \Cref{sec:oracle-constructions} does this for the NUCT.  The third depends only on the truncation rank $K$ and can be built once and reused.
\begin{enumerate}[label=(\roman*)]
  \item \emph{Node oracle.}  It writes the $m$-bit stored node $\tau_j$:
  \begin{equation}\label{eq:node_oracle}
    O_\tau\ket{j}\ketzero m =\ket{j}\ket{\tau_j}.
  \end{equation}
  \item \emph{Row-access oracle.}  It gives the positions of the nonzero entries in each row of the matrix $M_\sigma$ of \Cref{lem:routing-factorization}.  Their values are all $1$.  If row $\ell$ contains the ordered column indices $k_{\ell,0}<\cdots<k_{\ell,d_\ell-1}$ and $d_r\le N$ is a declared row-sparsity bound, define
  \begin{equation}
    \rho_{\ell,r}=
    \begin{cases}
      k_{\ell,r},&0\le r<d_\ell,\\
      N+r,&d_\ell\le r<d_r,
    \end{cases}
  \end{equation}
  where the padding values $N+r$ mark empty slots.  The oracle is
  \begin{equation}\label{eq:row_oracle}
    O_r\ket\ell\ket r=\ket\ell\ket{\rho_{\ell,r}}.
  \end{equation}
  \item \emph{Coefficient preparation.}  For each $r$, the inner preparation $\PREP_r$ and the outer preparation $\PREP_{\mathrm{out}}$ act as
  \begin{equation}
  \PREP_r\ketzero{\lceil\log_2K\rceil}
  =\frac{1}{\sqrt{\lambda_r}}
  \sum_{\ell=0}^{K-1}\sqrt{|\alpha'_{\ell r}|}\ket\ell,
  \qquad
  \PREP_{\mathrm{out}}\ketzero{\lceil\log_2K\rceil}
  =\frac{1}{\sqrt{\Lambda}}
  \sum_{r=0}^{K-1}\sqrt{\lambda_r}\ket r.
    \label{eq:paper-universal-prep}
  \end{equation}
  with $\lambda_r$ and $\Lambda$ as in \Cref{eq:lambda-r}.  The phases of $\alpha'_{\ell r}$ are applied in the SELECT step.
\end{enumerate}
\cite{AKY26} treats the coefficient preparation as an oracle (Assumption~18).  The next lemma shows that it does not depend on the transform instance, which is why we call these preparations \emph{universal}.

\begin{lemma}[Coefficient PREP are node-independent]
\label{lem:universal-coefficient-prep}
Fix the truncation rank $K$ and the window convention $\gamma=1/2$.  The
coefficients $\alpha'_{\ell r}$ of \Cref{eq:endpoint-adjusted-coeff}, the
branch weights $\lambda_r$ and total mass $\Lambda$ of
\Cref{eq:lambda-r}, and therefore the preparations \Cref{eq:paper-universal-prep} do not depend on the nodes $t$ or $\tau$, the nearest points $\sigma$, the
row multiplicities, the geometry parameter $\kappa$, or the input data $f$.
The same preparations serve the transposed construction for $F_t^{\mathsf T}$.
\end{lemma}

\begin{proof}
By \Cref{eq:bessel-coeff,eq:endpoint-adjusted-coeff}, $\alpha'_{\ell r}$ is
a function of $(\ell,r)$ and the fixed Bessel argument $-\pi/4$ determined by
$\gamma=1/2$ alone.  The weights $\lambda_r$ and $\Lambda$ are finite
sums of $|\alpha'_{\ell r}|$, and the amplitudes in
\Cref{eq:paper-universal-prep} are functions of these numbers.  The phases
$\arg(\alpha'_{\ell r})$ are carried by the select oracle and are likewise
instance independent.  In the factorization
\Cref{eq:dyadic-direct-factorization}, the node data enter only through
$T_\ell(z_j)$ in $D_{u_r}$ and through $M_\sigma$, and the frequency
coordinate enters only through $T_r(w_k)$ in $D_{v_r}$.  Transposition does
not change the coefficient table.
\end{proof}

\begin{remark}[Reusing the coefficient preparations]
\label{rem:universal-coefficient-prep}
We write $\PREP_{\mathrm{LCU}}$ for the inner preparations $\{\PREP_r\}_{r<K}$ together with the outer preparation $\PREP_{\mathrm{out}}$.  Their circuits are given in \Cref{prop:coefficient-prep,prop:outer-lcu-compilation}.  They correspond to Equation~(64) and Assumption~18 of \cite{AKY26}.  The symbol $\PREP_f$ is reserved for loading the input data.

By \Cref{lem:universal-coefficient-prep}, all rotation angles and phases of $\PREP_{\mathrm{LCU}}$ can be computed classically in advance.  The same circuit can therefore be reused in every NUQFT or NUCT call with the same $K$ and angle precision $b_\alpha$ (\Cref{prop:coefficient-prep}).  Each call still runs $\PREP_r$ and $\PREP_r^\dagger$, so reuse saves the classical precomputation, not the gate cost.
\end{remark}

\subsection{Conditioning-free NUQFT}
\label{sec:exact-dyadic-nuqft}

We now make the conditioning-free strategy of \Cref{sec:CFNUQFT} precise: the stored fixed-point nodes are treated as the exact nodes of the transform, and $\sigma_j$ and $z_j$ are computed exactly from them.

\begin{theorem}[Conditioning-free Type-II NUQFT]
\label{thm:dyadic-nuqft}
Let $N=2^q$ and let $t=(t_0,\ldots,t_{N-1})\in[0,1)^N$ define the
normalized Type-II NUDFT
\begin{equation}
  (F_t)_{kj}=\frac{1}{\sqrt N}e^{-2\pi i kt_j},
  \qquad 0\le j,k<N,
  \label{eq:proof-target-type-II}
\end{equation}
with rows indexed by the frequencies $k$ and columns by the nodes $t_j$.
This is \Cref{eq:type-II-mat} with the indices renamed.
Suppose that a node circuit supplies $m$-bit fixed-points $\tau=(\tau_0,\ldots,\tau_{N-1})$ and that the nearest-grid
matrix $M_\sigma$ determined by $\tau$ has ordered sparse row access with a
declared row-sparsity bound $d_r$.  

For every $0<\eps\le1$, there is a quantum circuit $\matU_{\mathrm{II}}$ that implements the block encoding of $F_t$ with normalization
$O(\sqrt{d_r})$ and total operator-norm error
\begin{equation}
  \| \matU_{\mathrm{II}} - F_t \|
  \le
  \eps
  +\frac{2\pi}{\sqrt3}N^{3/2}\norm{t-\tau}_\infty.
  \label{eq:conditioning-free-main-total-error}
\end{equation}
Given the oracles of \Cref{sec:nuqft-oracles}, the quantum circuit uses
\begin{equation}
  O(m+L)\ \text{qubits},
  \qquad
  \widetilde O(m^2+L^2)\ \text{logical gates},
  \label{eq:conditioning-free-main-complexity}
\end{equation}
where $L:=q+\log(1/\eps)$.
\end{theorem}

\begin{proof} 
The proof has two parts.  The error analysis bounds the three terms of \Cref{eq:error-budget}, with $F_\tau$ and $F_\tau^{(K)}$ as in \Cref{eq:proof-computational-type-II,eq:dyadic-direct-factorization}, one bound per term.  The resource analysis then counts qubits and gates.

\begin{itemize}

\item \emph{Bound 1: truncation error $\norm{F_\tau-F_\tau^{(K)}}$.}  By \Cref{eq:background-phase-split,eq:residual-kernel}, $e^{-2\pi ik\tau_j}=e^{-2\pi ik\sigma_j/N}e^{-i\pi z_j/2}H(w_k,z_j)$ exactly.  Since $(FM_\sigma)_{kj}=F_{k\sigma_j}=e^{-2\pi ik\sigma_j/N}/\sqrt N$, the $(k,j)$ entry of $F_\tau^{(K)}$ is $\sum_rv_r(k)F_{k\sigma_j}u_r(j)$, that is,
\begin{equation}
  (F_\tau^{(K)})_{kj}
  =\frac{e^{-2\pi i k\sigma_j/N}e^{-i\pi z_j/2}}{\sqrt N}
    H_K(w_k,z_j ),\quad \text{and hence}\quad |(F_\tau-F_\tau^{(K)})_{kj}|
  \le\frac{\Delta_K}{\sqrt N}.
\end{equation}
Consequently,
\begin{equation}
  \norm{F_\tau-F_\tau^{(K)}}
  \le\norm{F_\tau-F_\tau^{(K)}}_{\rm F}
  \le\sqrt N\,\Delta_K.
  \label{eq:nuqft-truncation-operator-bound}
\end{equation}
By part~(i) of
\Cref{lem:bivariate-chebyshev-bessel} with $\delta=\eps/2$, choosing $K=\lceil \log_2 \left( \frac{48\sqrt{N}}{\epsilon}\right) \rceil =  O\!\left(q+\log\frac1\eps\right)$
ensures that the first term of \Cref{eq:error-budget} satisfies
\[
\norm{F_\tau-F_\tau^{(K)}}
  \le\sqrt N\Delta_K\le\eps/2 .
\]

\item \emph{Bound 2: implementation error $\norm{\matU_{\mathrm{II}}-F_\tau^{(K)}}$.}  The QFT and the block encoding of $M_\sigma$ (\Cref{lem:routing-factorization}) are exact, so $\matU_{\mathrm{II}}$ encodes $\sum_r\widetilde A_r$ with $\widetilde A_r$ as in \Cref{lem:assembled-factor-error}, while $F_\tau^{(K)}=\sum_rA_r$.  With $p$ as in \Cref{eq:lemma38-explicit-precision-choice}, part~(iii) of that lemma gives
\begin{equation}
  \norm{\matU_{\mathrm{II}}-F_\tau^{(K)}}\le\frac{\eps}{2}.
  \label{eq:nuqft-implementation-bound}
\end{equation}
Here $p=O(\log K+\tfrac12\log d_r+\log(1/\eps))=O(q+\log(1/\eps))$, using $d_r\le N$.

Overall, combining \Cref{eq:nuqft-truncation-operator-bound} and
\Cref{eq:nuqft-implementation-bound} yields
\begin{equation}
  \norm{\matU_{\mathrm{II}}-F_\tau} \le \norm{\matU_{\mathrm{II} }- F_\tau^{(K)}} + \norm{F_\tau^{(K)} - F_\tau }
  \le\eps.
  \label{eq:proof-internal-fixed-point-error}
\end{equation}

\item \emph{Bound 3: node error $\norm{F_\tau-F_t}$.}  This is the error from using the stored nodes $\{\tau_j\}$ instead of the target nodes $\{t_j\}$, the third term of \Cref{eq:error-budget}.  
For every $j,k$, the inequality $|e^{iu}-e^{iv}|\le|u-v|$ gives
\begin{align*}
  |(F_t-F_\tau)_{kj}|
  =\frac1{\sqrt N} |e^{-2\pi ikt_j}-e^{-2\pi ik\tau_j}| 
  \le\frac{2\pi k}{\sqrt N}|t_j-\tau_j| 
  \le\frac{2\pi k}{\sqrt N}\norm{t-\tau}_\infty.
\end{align*}
Therefore,
\begin{align}
  \norm{F_t-F_\tau}
  &\le\norm{F_t-F_\tau}_{\rm F} \notag\\
  &\le
  \left(
    \sum_{j=0}^{N-1}\sum_{k=0}^{N-1}
    \frac{4\pi^2k^2}{N}\norm{t-\tau}_\infty^2
  \right)^{1/2} \notag\\
  &=2\pi\left(\sum_{k=0}^{N-1}k^2\right)^{1/2}
    \norm{t-\tau}_\infty \notag\\
  &\le\frac{2\pi}{\sqrt3}N^{3/2}\norm{t-\tau}_\infty.
  \label{eq:operator-node-surrogate-error}
\end{align}
Together with \Cref{eq:proof-internal-fixed-point-error}, this proves
\begin{equation}
  \norm{\matU_{\mathrm{II}}-F_t}
  \le
  \eps+\frac{2\pi}{\sqrt3}N^{3/2}\norm{t-\tau}_\infty,
\end{equation}
which completes the proof for \Cref{eq:conditioning-free-main-total-error}.
\end{itemize}

\paragraph{Resource analysis.}
Let $L:=q+\log(1/\eps)$.  The choices above give $K=O(L)$ and $p=O(L)$.
The node register has $m$ bits.  Computing $\sigma_j$ and $z_j$ exactly uses $O(m+q)$
temporary bits and $O(m+q)$ reversible arithmetic gates (\Cref{lem:routing-factorization}).  The node factor keeps the offset $z_j$, including its sign and the carry at the last grid cell, at width $w_{\rm ac}:=\max\{m+2,p\}+g_{\rm ac}$ with $g_{\rm ac}=O(\log(m+p+2))$, so $w_{\rm ac}=O(m+p)$.  Its inverse-trigonometric and phase arithmetic costs $\widetilde O(w_{\rm ac}^2)$ gates and $O(w_{\rm ac})$ work qubits by \Cref{thm:reversible-elementary-functions}.  The frequency factor acts on the exact $q$-bit input $w_k$ and costs $\widetilde O((q+p)^2)$.  Multiplying the angle by the coherent degree register $0\le n<K$ adds $\widetilde O((w_{\rm ac}+\log K)^2)$, of the same order since $p\ge\log_2K$, and the phase $e^{-i\pi z_j/2}$ uses $O(m)$ controlled phases.  The QFT contributes $O(q^2)$
controlled rotations, and the selector registers use $O(\log K+\log d_r)$ qubits.
Thus, before compiling the universal coefficient states, the circuit uses
\begin{equation}
  O(m+L)\ \text{qubits},
  \qquad
  \widetilde O(m^2+L^2)\ \text{logical gates}.
  \label{eq:dyadic-nuqft-resources}
\end{equation}

For fixed $r$, M\"ott\"onen synthesis \cite{MVBS05} (\Cref{thm:mottonen-state-preparation}) prepares the $K$-dimensional
coefficient state using $O(K)$ elementary rotations and CNOTs.  Coherently
multiplexing the $K$ such preparations costs $O(K^2)=O(L^2)$ gates, and the
weighted outer state costs $O(K)=O(L)$ gates.  These additions do not change
\Cref{eq:dyadic-nuqft-resources}.  The costs of the instance-specific node
and ordered row-access circuits remain separate, as stated in the theorem.
If $m=O(L)$, the displayed bounds reduce to $O(L)$ qubits and
$\widetilde O(L^2)$ logical gates.
\end{proof}

\paragraph{Note.} Because $F^{\mathsf T}=F$ and diagonal matrices are symmetric, the transposed factors $A_r^{\mathsf T}=D_{u_r}M_\sigma^{\mathsf T}FD_{v_r}$ block-encode $F_t^{\mathsf T}$, which is $F_{\mathrm{II}}/\sqrt N$ in the orientation of \cite{AKY26} (\Cref{rem:orientation}).  The spectral norm is invariant under transposition, and the block encoding of $M_\sigma^{\mathsf T}$ is the adjoint of that of $M_\sigma$ (\Cref{lem:routing-factorization}), so the same bounds, oracles and resources apply.

\section{Application to the Non-Uniform Chebyshev Transform}
\label{sec:nupct}

We now apply \Cref{thm:dyadic-nuqft} to the NUCT defined in \Cref{eq:nupct-def}.  The argument has three steps.  \Cref{sec:nupct-reduction} writes $C_N$ as the average of two NUDFTs whose nodes are the angles of the grid points.  \Cref{sec:oracle-constructions} builds the ingredients that \Cref{thm:dyadic-nuqft} takes as given: the node oracle, row access to $M_\sigma$, and the coefficient preparation.  \Cref{sec:end-to-end} puts them together.

Throughout, $N=2^q$ with $q\ge2$, and for $0\le k<N$,
\begin{equation}
  x_k=-1+\frac{2k}{N},
  \qquad
  \theta_k=\arccos x_k\in(0,\pi],
  \qquad
  t_k=\frac{\theta_k}{2\pi}\in\left(0,\tfrac12\right].
  \label{eq:angular-nodes}
\end{equation}
The nodes $t_k$ are irrational in general.  Only the endpoint $t_0=1/2$ is exact.  Since $x_k$ increases with $k$ and $\arccos$ is decreasing, $t_k$ decreases with $k$.

\subsection{Reduction to two NUDFTs}
\label{sec:nupct-reduction}

\begin{theorem}[NUCT as an average of two NUDFTs]
\label{thm:nupct-reduction}
Let $C_N$ be as in \Cref{eq:C-matrices}, let $F_t$ be the Type-II matrix \Cref{eq:type-II-mat} with the nodes \Cref{eq:angular-nodes}, and let $t_k^-=(-t_k)\bmod1$.  Then
\begin{equation}
  C_N=\frac{1}{2}\left(F_t+F_{t^-}\right).
  \label{eq:nupct-average}
\end{equation}
\end{theorem}

\begin{proof}
Since $T_j(\cos\theta)=\cos(j\theta)$, $T_j(x_k)=\cos(j\theta_k)=\frac12(e^{-ij\theta_k}+e^{ij\theta_k})$.  Since $\theta_k=2\pi t_k$ and $e^{-2\pi ijt_k^-}=e^{2\pi ijt_k}$ for integer $j$,
\begin{equation*}
  (C_N)_{jk}
  =\frac{1}{2\sqrt N}
   \left(e^{-2\pi ijt_k}+e^{-2\pi ijt_k^-}\right)
  =\frac12\left((F_t)_{jk}+(F_{t^-})_{jk}\right).
\end{equation*}
\end{proof}

This is the reduction in the proof of \cite[Corollary~2.2(2)]{DHR97}, with the nodes $e^{\pm i\theta_k}$ on the unit circle.  We apply the factor $\frac12$ as a two-term LCU (\Cref{alg:nupct}) rather than folding it into conjugate-symmetric data.

A circuit cannot store the irrational nodes $t_k$, so it works with $m$-bit approximations $\tau_k$ and implements
\begin{equation}
  C_\tau:=\frac12\left(F_{\tau}+F_{\tau^-}\right),
  \qquad
  \norm{t-\tau}_\infty=\max_k|\tau_k-t_k|.
  \label{eq:quantized-nupct-matrix}
\end{equation}
The difference between $C_\tau$ and $C_N$ is the node term of \Cref{thm:dyadic-nuqft}.  For each of the two NUDFTs, that theorem needs a node oracle \Cref{eq:node_oracle}, row access \Cref{eq:row_oracle} to $M_\sigma$ with a small row-sparsity bound $d_r$, and the coefficient preparations.  The next subsection constructs all of them.  The negative nodes $\tau^-$ come almost for free from $\tau$ (\Cref{lem:negative-branch}).

\subsection{Constructing the Oracles}
\label{sec:oracle-constructions}

This subsection builds three circuits.  The node oracle $O_\tau$ returns an $m$-bit approximation $\tau_k$ of $t_k$ (\Cref{prop:angle-oracle}).  The row oracle $O_r$ gives row access to $M_\sigma$ for both $\tau$ and $\tau^-$ (\Cref{prop:row-oracle,lem:negative-branch}).  The coefficient preparations $\PREP_{\mathrm{LCU}}$ prepare the Bessel-coefficient states (\Cref{prop:coefficient-prep,prop:outer-lcu-compilation}).  The input state is prepared by a given circuit $\PREP_f\ketzero q=\ket f$, as in \Cref{sec:problem-statement}.  All three use standard reversible arithmetic and state preparation, so we state the results here and give the circuits and proofs in \Cref{app:oracle-constructions}.

\paragraph{Node oracle.}  Since $x_k=(2k-N)/N$ is an exact fixed-point number, the node oracle evaluates $\arccos$ on an exact input and rounds only once, at the end (\Cref{alg:angle-oracle}).

\begin{proposition}[Node oracle]
\label{prop:angle-oracle}
Let $m\ge q+1$.  There is a reversible circuit
$O_\tau\ket k\ketzero m=\ket k\ket{\tau_k}$, where $\tau_k\in[0,1/2]$ is an $m$-bit number with
\begin{equation}
  \abs{\tau_k-t_k}\le 2^{-m}.
  \label{eq:angle-error}
\end{equation}
It uses $O(m)$ qubits, $\widetilde O(m^2)$ gates, and depth $\widetilde O(m)$.
\end{proposition}

This is the only $\arccos$ evaluation which contributes to error in analysis and that error enters only through the node term $\norm{t-\tau}_\infty$ of \Cref{eq:conditioning-free-main-total-error}.  The two $\arccos$ evaluations inside the improved NUQFT, for the frequency and node factors, act on exact inputs (\Cref{lem:diagonal-factor-implementation}), so no $\kappa$ appears.

\begin{lemma}[Number of bits per node]
\label{lem:node-precision}
Let $0<\eps\le1$ and
\begin{equation}
  m=
  \left\lceil\frac32q+\log_2\frac{24\pi}{\eps}\right\rceil.
  \label{eq:unconditional-node-bits}
\end{equation}
Then the node oracle of \Cref{prop:angle-oracle} satisfies
\begin{equation}
  \frac{2\pi}{\sqrt3}N^{3/2}\norm{t-\tau}_\infty
  \le\frac{\eps}{12\sqrt3}<\frac{\eps}{3}
  \qquad\text{and}\qquad
  N\norm{t-\tau}_\infty<\frac1{2\pi}.
  \label{eq:explicit-node-budget}
\end{equation}
In particular $m=O(L)$ with $L=q+\log(1/\eps)$.
\end{lemma}

\paragraph{Row access.}  \Cref{thm:dyadic-nuqft} also needs ordered access to the nonzero entries in each row of the matrix $M_\sigma=\sum_k\ket{\sigma_k}\!\bra k$ of \Cref{eq:prelim-routing-matrix-data}, where $\sigma_k$ is the grid point of $\tau_k$ (\Cref{eq:dyadic-routing-data-a}).  Since $\tau_k\le1/2$, $s_k^{\mathrm{unw}}\le N/2$, so no wrap-around occurs and $\sigma_k=s_k^{\mathrm{unw}}$.  The grid-point oracle $O_\sigma\ket k\ket0=\ket k\ket{\sigma_k}$ costs one call to $O_\tau$, a rounding, and an uncomputation, and column access to $M_\sigma$ is $O_\sigma$ itself.  Three facts make row access cheap: each row has at most five entries, they are consecutive integers, and they lie in a short window that can be computed from the row number.

\begin{lemma}[Rows of $M_\sigma$]
\label{lem:quantized-multiplicity}
Suppose $N\norm{t-\tau}_\infty<1/(2\pi)$.  Rows $\ell>N/2$ of $M_\sigma$ are empty.  For $0\le\ell\le N/2$, let
\begin{equation}
  \Theta_{\ell,-}=\max\!\left\{0,
    \frac{2\pi}{N}\left(\ell-\frac34\right)\right\},
  \qquad
  \Theta_{\ell,+}=\min\!\left\{\pi,
    \frac{2\pi}{N}\left(\ell+\frac34\right)\right\}.
  \label{eq:expanded-angle-window}
\end{equation}
Then every $k$ with $\sigma_k=\ell$ satisfies
\begin{equation}
  a_\ell\le k\le b_\ell,
  \qquad
  a_\ell=\frac N2(1+\cos\Theta_{\ell,+}),
  \qquad
  b_\ell=\frac N2(1+\cos\Theta_{\ell,-}),
  \label{eq:expanded-index-window}
\end{equation}
and $b_\ell-a_\ell<5$.  In particular, every row has at most five entries.  Moreover, $\sigma_k$ is non-increasing in $k$, so the entries of each row are consecutive integers.
\end{lemma}

\begin{proposition}[Row access from eight candidates]
\label{prop:row-oracle}
Suppose $N\norm{t-\tau}_\infty<1/(2\pi)$ and set $d_r=\min\{5,N\}$.  The row-access oracle $O_r$ of \Cref{eq:row_oracle} for $M_\sigma$ can be implemented with at most sixteen calls each to $O_\sigma$ and $O_\sigma^\dagger$, $\widetilde O(m^2+q^2)$ gates, and $O(m+q)$ work qubits.  It needs no lookup table and no search.
\end{proposition}

\paragraph{Negative nodes.} When calculating offsets, there is a possibility of negative valued nodes, we show below that it is easy to extend NUQFT for negative nodes.

\begin{lemma}[Negative nodes from the positive ones]
\label{lem:negative-branch}
Suppose $N\norm{t-\tau}_\infty<1/(2\pi)$.  Then $\tau_k\in(0,1/2]$, so $\tau_k^-=1-\tau_k$.  Let $s_k^{\mathrm{unw}}$, $\sigma_k$ and $z_k$ be as in \Cref{eq:dyadic-routing-data-a} for $\tau_k$, and set
\begin{equation}
  s_k^-=N-s_k^{\mathrm{unw}},
  \qquad
  z_k^-=-z_k,
  \qquad
  \sigma_k^-=(-\sigma_k)\bmod N.
  \label{eq:signed-label-reflection}
\end{equation}
Then $N\tau_k^-=s_k^-+z_k^-/2$ with $z_k^-\in[-1,1]$, and row $\ell$ of $M_{\sigma^-}$ is row $(-\ell)\bmod N$ of $M_\sigma$.  Consequently:
\begin{enumerate}[label=(\roman*)]
  \item The NUQFT for $\tau^-$ is the circuit for $\tau$ with $(\tau_k,s_k^{\mathrm{unw}},z_k)$ replaced by $(1-\tau_k,s_k^-,-z_k)$, which costs $O(m+q)$ extra gates and no second $\arccos$.
  \item Its row oracle is obtained by negating $\ell$ modulo $N$, applying the row oracle of \Cref{prop:row-oracle}, and negating back, with the same bound $d_r$.
\end{enumerate}
\end{lemma}

\paragraph{Coefficient preparation.}  The coefficient preparations do not depend on the nodes (\Cref{lem:universal-coefficient-prep}), so the constructions below apply to any NUQFT and are compiled once for a given $K$ and precision.  For fixed $r$, with $\lambda_r$ as in \Cref{eq:lambda-r}, define
\begin{equation}
  \ket{\chi_r}
  =\frac{1}{\sqrt{\lambda_r}}
   \sum_{\ell=0}^{K-1}
   \sqrt{|\alpha'_{\ell r}|}\ket\ell,
  \qquad
  \phi_{\ell r}=\arg(\alpha'_{\ell r}).
  \label{eq:chi-r}
\end{equation}
We implement
\begin{equation}
  \PREP_r\ketzero{\lceil\log_2K\rceil}=\ket{\chi_r}
  \label{eq:prep-r}
\end{equation}
and apply the phase $e^{i\phi_{\ell r}}$ inside the SELECT operation, so that the LCU produces the complex coefficient $\alpha'_{\ell r}$.  For the coefficients \Cref{eq:bessel-coeff} the Bessel factors are real, so $\phi_{\ell r}$ is a multiple of $\pi/2$ and the phases cost only Clifford gates.
Both preparations use the M\"ott\"onen state-preparation circuit \cite{MVBS05} (\Cref{thm:mottonen-state-preparation}).

\begin{proposition}[Inner LCU coeff preparation]
\label{prop:coefficient-prep}
Let $\alpha$ be the normalization of the block encoding in which the preparations are used.  After classically computing the $K^2$ coefficients and the angles \Cref{eq:mottonen-angle} to accuracy $O(2^{-b_\alpha})$, each state $\ket{\chi_r}$ can be prepared with $O(K)$ rotations and CNOT gates on $O(\log K)$ qubits, and the controlled family
\begin{equation}
  \PREP_{\mathrm{LCU}}
  =\sum_{r=0}^{K-1}\ket r\!\bra r\otimes\PREP_r
  \label{eq:multiplexed-prep}
\end{equation}
with $O(K^2)$ one- and two-qubit gates.  Choosing
\begin{equation}
  b_\alpha=O\!\left(
  \log\frac{K^2\alpha}{\eps_{\mathrm{prep}}}
  \right)
  \label{eq:prep-precision}
\end{equation}
keeps the resulting error in the encoded matrix $O(\eps_{\mathrm{prep}})$.
\end{proposition}

\begin{proposition}[Outer LCU preparations]
\label{prop:outer-lcu-compilation}
Let the $r$th term of the outer LCU have normalization $\beta_r=\sqrt{d_r}\lambda_r$ (\Cref{lem:assembled-factor-error}), and let $\beta_\star=\sqrt{d_r}\Lambda$.  After classically computing the weights $\lambda_r$:
\begin{enumerate}[label=(\alph*),leftmargin=2.2em]
  \item The uniform outer LCU prepares $K^{-1/2}\sum_{r<K}\ket r$ and, in branch $r$, rotates an ancilla so that its $\ket0$ amplitude is $\beta_r/\beta_\star=\lambda_r/\Lambda$.  Its normalization is
  \begin{equation}
    \alpha_{\mathrm{unif}}=K\sqrt{d_r}\,\Lambda.
    \label{eq:uniform-outer-normalization}
  \end{equation}
  \item The weighted outer LCU prepares
  \begin{equation}
    \PREP_{\mathrm{out}}\ketzero{\lceil\log_2K\rceil}
    =\frac{1}{\sqrt{\Lambda}}
     \sum_{r=0}^{K-1}\sqrt{\lambda_r}\ket r,
    \label{eq:outer-prep}
  \end{equation}
  and its normalization is
  \begin{equation}
    \alpha_{\mathrm{wt}}=\sqrt{d_r}\,\Lambda.
    \label{eq:weighted-normalization}
  \end{equation}
\end{enumerate}
Each preparation, and the branch-dependent rotation in (a), uses $O(K)$ rotations and CNOT gates.  Choosing
\begin{equation}
  b_{\mathrm{out}}
  =O\!\left(
      \log\frac{K\alpha_{\mathrm{unif}}}{\eps_{\mathrm{out}}}
    \right)
  \label{eq:outer-prep-precision}
\end{equation}
bits of angle precision keeps the resulting error in the encoded matrix $O(\eps_{\mathrm{out}})$.
\end{proposition}

\subsection{Main result}
\label{sec:end-to-end}

Let $V_+$ be the block encoding of $F_{\tau}$ given by \Cref{thm:dyadic-nuqft} with the oracles above, and $V_-$ the one for $F_{\tau^-}$ obtained from it by \Cref{lem:negative-branch}.  Both have the same normalization $\alpha$ and the same number $a$ of ancillas.  Define
\begin{equation}
  \SEL_{\pm}
  =\ket0\!\bra0\otimes V_+
   +\ket1\!\bra1\otimes V_-,
  \qquad
  V_{\mathrm{NUCT}}
  =(H\otimes I)\SEL_{\pm}(H\otimes I),
  \label{eq:nupct-circuit}
\end{equation}
where the Hadamards act on a sign qubit.  By \Cref{lem:negative-branch}, $\SEL_\pm$ is the single circuit $V_+$ with its node data controlled on the sign qubit, not two separate circuits.

\begin{algorithm}[H]
\caption{Quantum NUCT}
\label{alg:nupct}
\begin{algorithmic}[1]
\Require $N=2^q$, target error $\eps$, the oracles of \Cref{sec:oracle-constructions}, and a circuit $\PREP_f$ with $\PREP_f\ketzero q=\ket f$
\State Prepare $\ket f=\PREP_f\ketzero q$
\State Initialize the NUQFT ancillas and a sign qubit to $\ket0$
\State Apply $H$ to the sign qubit
\State Apply $\SEL_\pm$: the NUQFT for $\tau$ or $\tau^-$, controlled on the sign qubit
\State Apply $H$ to the sign qubit
\State Measure the sign qubit and the NUQFT ancillas, and accept only the all-zero outcome
\State \Return the system register, in a state proportional to $\widetilde C_N\ket f$
\end{algorithmic}
\end{algorithm}

\begin{theorem}[Quantum NUCT]
\label{thm:compiled-nupct}
Let $N=2^q$ with $q\ge2$, let $0<\eps\le1$, and let $L=q+\log(1/\eps)$.  With $m$ as in \Cref{lem:node-precision}, the circuit $V_{\mathrm{NUCT}}$ is a block encoding of $C_N$ with normalization
\begin{equation}
  \alpha=\sqrt{d_r}\,\Lambda<24\sqrt5=O(1)
  \label{eq:alpha-K}
\end{equation}
and error at most $\eps$.  It uses $O(L)$ qubits and
\begin{equation}
  \widetilde O(L^2)\ \text{reversible-arithmetic, Clifford, Toffoli, and controlled-rotation gates}.
  \label{eq:compiled-gates}
\end{equation}
Synthesizing every rotation over Clifford+$T$ gives
\begin{equation}
  \widetilde O(L^3)\ \text{Clifford+$T$ gates}.
  \label{eq:compiled-clifford-t}
\end{equation}
No parameter depending on the node geometry appears.  These counts exclude $\PREP_f$ and amplitude amplification.  With the uniform outer LCU instead, the normalization is $O(K)=O(L)$.
\end{theorem}

\begin{proof}
We split the error budget into three parts of $\eps/3$.

\emph{Node error.}  By \Cref{lem:node-precision}, $(2\pi/\sqrt3)N^{3/2}\norm{t-\tau}_\infty<\eps/3$.  Since $t_k,\tau_k>0$ (\Cref{lem:negative-branch}), $|\tau_k^--t_k^-|=|\tau_k-t_k|$.  Applying \Cref{eq:operator-node-surrogate-error} to both terms of \Cref{eq:quantized-nupct-matrix} gives $\norm{C_\tau-C_N}\le(2\pi/\sqrt3)N^{3/2}\norm{t-\tau}_\infty<\eps/3$.

\emph{NUQFT error.}  By \Cref{lem:node-precision}, $N\norm{t-\tau}_\infty<1/(2\pi)$, so \Cref{prop:row-oracle,lem:negative-branch} give row access for both node sets with $d_r=\min\{5,N\}$.  \Cref{thm:dyadic-nuqft}, applied with accuracy $\eps/3$ to the stored nodes themselves (so that its node term vanishes), gives $V_\pm$ encoding matrices $\widetilde F_\pm$ with $\norm{\widetilde F_+-F_{\tau}}\le\eps/3$ and $\norm{\widetilde F_--F_{\tau^-}}\le\eps/3$.  With the weighted outer LCU, their normalization is $\sqrt{d_r}\Lambda<24\sqrt5$ by \Cref{lem:assembled-factor-error,lem:constant-bessel-mass}.  Projecting the sign qubit of \Cref{eq:nupct-circuit} onto $\ket0$ gives $\frac12(V_++V_-)$, so $V_{\mathrm{NUCT}}$ encodes $\widetilde C_N=\frac12(\widetilde F_++\widetilde F_-)$ with the same normalization and $\norm{\widetilde C_N-C_\tau}\le\eps/3$.

\emph{Rotation precision.}  The last $\eps/3$ covers finite-precision rotations.  Choosing $\eps_{\mathrm{prep}}$ and $\eps_{\mathrm{out}}$ in \Cref{prop:coefficient-prep,prop:outer-lcu-compilation} as suitable constant fractions of $\eps$ keeps the error from the coefficient angles below $\eps/6$.  The Clifford+$T$ synthesis below contributes at most another $\eps/6$.  The three parts give $\norm{\widetilde C_N-C_N}\le\eps$.

\emph{Resources.}  Since $m=O(L)$, \Cref{thm:dyadic-nuqft} uses $O(L)$ qubits and $\widetilde O(L^2)$ gates, including $O(K^2)=O(L^2)$ gates for the coefficient preparations.  Its oracles add $\widetilde O(m^2)$ gates for $O_\tau$ (\Cref{prop:angle-oracle}), $\widetilde O(m^2+q^2)$ for each of the constant number of row-oracle calls made by the sparse-access block encoding (\Cref{prop:row-oracle}), and $O(m+q)$ for the negative branch (\Cref{lem:negative-branch}).  The sign qubit adds one qubit and two Hadamards.  This proves \Cref{eq:compiled-gates}.  For Clifford+$T$, let $R=\widetilde O(L^2)$ be the number of rotations and synthesize each to error $\eps/(6\alpha R)$, which costs $O(\log(\alpha R/\eps))=O(L)$ $T$ gates \cite{RossSelinger2016} and adds at most $\eps/6$ to the encoded matrix.  This gives \Cref{eq:compiled-clifford-t}.
\end{proof}
\Cref{thm:compiled-nupct} gives a block encoding, which embeds $C_N/\alpha$ in a unitary but does not yet produce the state \Cref{eq:quantum-task}.  Applying it to $\ket f$ yields the desired state only when the ancillas are measured in the all-zero outcome.  The corollary below turns the block encoding into a solution of the state-preparation task.  It gives the probability of that outcome, the number of amplitude-amplification rounds needed to boost it, and the accuracy of the normalized output state.  The matrix error $\eps$ is absolute, while the output state is normalized by $r=\norm{C_N\ket f}$.  The state error is therefore relative, and $\eps$ must be small compared to $r$.
\begin{corollary}[Output state]
\label{cor:output-state}
Let $\ket f=\norm f_2^{-1}\sum_kf_k\ket k$ be the input state, let $\widetilde C_N$ and $\alpha$ be as in \Cref{thm:compiled-nupct}, and let $r:=\norm{C_N\ket f}=\norm c_2/(\sqrt N\norm f_2)$, where $c$ is as in \Cref{eq:nupct-def}.  Suppose $r>\eps$.  Then \Cref{alg:nupct} succeeds with probability $\norm{\widetilde C_N\ket f}^2/\alpha^2\ge(r-\eps)^2/\alpha^2$, so amplitude amplification needs $O(\alpha/(r-\eps))$ uses of $V_{\mathrm{NUCT}}$, $\PREP_f$ and their inverses.  The output state satisfies
\begin{equation}
  \norm{
  \frac{\widetilde C_N\ket f}{\norm{\widetilde C_N\ket f}}
  -\frac{c}{\norm c_2}}
  \le \frac{2\eps}{r-\eps}.
  \label{eq:normalized-state-error}
\end{equation}
\end{corollary}

\begin{proof}
The success probability is the squared norm of the encoded block applied to $\ket f$, divided by $\alpha^2$, and $\norm{\widetilde C_N\ket f}\ge r-\eps$.  For \Cref{eq:normalized-state-error}, use $C_N\ket f/r=c/\norm c_2$ and the elementary bound $\norm{u/\norm u-v/\norm v}\le2\norm{u-v}/\norm u$ with $u=\widetilde C_N\ket f$ and $v=C_N\ket f$.
\end{proof}

Since $\alpha=O(1)$, the number of amplification rounds is governed by $r$, which depends on the data.  The transform itself costs $\widetilde O(L^2)$ gates, but an end-to-end application can still be dominated by loading $f$, by the factor $1/r$, or by reading out many coefficients.

\section{Conclusion and Discussion}
\label{sec:conclusion}
We gave a conditioning-free analysis of the NUQFT and used it to construct the non-uniform Chebyshev transform on equispaced nodes. The main consequence of removing $\kappa$ is that the cost of the NUQFT can now be stated in advance for any node set.  Previously, obtaining a polylogarithmic cost required lower-bounding the distance of every scaled node from a half-integer.  For structured nodes such as ours this becomes a number-theoretic question that we could not resolve.  Our bounds depend only on quantities that are easy to check, namely the bits per node, the target accuracy and the number $d_r$ of nodes that share a grid point.  The NUQFT can therefore be used as a black-box subroutine whose complexity guarantee does not require analyzing the nodes.  The remaining dependence on the geometry is through $d_r$, which counts collisions rather than measuring conditioning.  It is small whenever the nodes are spread out at the scale of the grid, and it grows only when many nodes cluster in one grid cell.  The constant bound on the total Bessel weight also keeps the normalization independent of the truncation rank.  Higher accuracy therefore costs only logarithmically many extra gates and no extra rounds of amplitude amplification.

For the Chebyshev transform, the result removes the restriction to Chebyshev nodes, hence it can be achieved for any set of arbitrary nodes.  Data from measurements and simulations is usually sampled on a uniform grid, and QFT-based quantum Chebyshev transforms would first require resampling it at Chebyshev nodes.  Our transform acts on uniformly sampled data directly.  Its output must be interpreted as the projections of the data onto the sampled Chebyshev polynomials, not the coefficients of an interpolating expansion, and recovering those coefficients is exponentially ill-conditioned on equispaced nodes.  The transform is therefore most useful where projections suffice, for example as the input to further polynomial transforms or other quantum linear-algebra routines.  As for any such primitive, an end-to-end speedup also depends on preparing the input state, on the norm $\norm{C_N\ket f}$ of the output, and on how much of the output must be read out.

Our original motivation was a quantum version of the Driscoll--Healy--Rockmore framework, which computes discrete polynomial transforms for any orthogonal polynomial family with a three-term recurrence and uses the Chebyshev transform as its first stage.  The main open question is whether the remaining stages, which convert Chebyshev projections into projections onto a general family through the recurrence, can also be implemented with polylogarithmic cost.   General quadrature weights $w_k$ should require only one more diagonal block encoding.  Other non-uniform node families can use the same NUQFT, and the work there is to bound $d_r$ and to construct the row oracle, as we did for equispaced nodes.

\section*{Acknowledgment}

\paragraph{AI Disclosure.} The authors came up with the novel idea for using NUQFT for Chebyshev Transform and improving the NUQFT by removing the condition number. We had a working proof and used ChatGPT-Sol-5.6 to refine the conditioning-free idea and oracle constructions. We used Claude Opus 5.5 for coming up with efficient Oracle constructions, based on standard techniques and help with writing the paper. The authors verified the correctness and originality of all content including references and take full responsibility for the work presented.

\bibliographystyle{alpha}
\bibliography{refs}

\appendix

\section{Oracle constructions}
\label{app:oracle-constructions}

This appendix gives the circuits and proofs for the oracles stated in \Cref{sec:oracle-constructions}.

\subsection{Node oracle}
\label{app:angle-oracle}

Since $x_k=(2k-N)/N$ is an exact fixed-point number, the node oracle evaluates $\arccos$ on an exact input and rounds only once, at the end.  Fix a constant number $g$ of guard bits, set the working precision $w=m+g$, and let $[1/(2\pi)]_w$ denote $1/(2\pi)$ rounded to $w$ bits.  The oracle computes
\begin{equation}
  \tau_k=
  \begin{cases}
    1/2,&k=0,\\
    \min\bigl\{\max\{\widetilde t_k,0\},1/2\bigr\},&k\ge1,
  \end{cases}
  \qquad
  \widetilde t_k:=\operatorname{round}_m\bigl([1/(2\pi)]_w\,\widetilde\theta_k\bigr),
  \label{eq:node-oracle-output}
\end{equation}
where $\widetilde\theta_k$ is the $w$-bit approximation of $\arccos x_k$ from part~(i) of \Cref{thm:reversible-elementary-functions} and $\operatorname{round}_m$ rounds to $m$ fractional bits.  \Cref{alg:angle-oracle} lists the steps of $O_\tau$ and the state after each one.

\begin{algorithm}[H]
\caption{Node oracle $O_\tau$}
\label{alg:angle-oracle}
\renewcommand{\algorithmicrequire}{\textbf{Given:}}
\begin{algorithmic}[1]
\Require $\ket k\ket0_{\rm f}\ket0_{\rm w}\ket0_{\rm o}$ with index $k$ ($q$ qubits), flag ($1$ qubit), work register, and output ($m$ qubits), with output precision $m\ge q+1$ and working precision $w=m+g$
\State \emph{Index and flag.}  Compute the signed integer $X_k=2k-N$ ($q+1$ qubits) by a shift and a subtraction, and $b_k=[k=0]$ by a comparison.
\State \emph{Angle.}  Compute $\widetilde\theta_k\approx\arccos(X_k/N)\in[0,\pi]$ ($w+2$ qubits: $2$ integer, $w$ fractional) by \Cref{thm:reversible-elementary-functions}, whose $O(q+w)$ scratch qubits are cleaned within this step.
\State \emph{Scale and round.}  Compute $P_k=[1/(2\pi)]_w\,\widetilde\theta_k\in[0,1)$ ($w$ qubits) by shifted additions, since the classical constant $[1/(2\pi)]_w$ needs no register.  Then compute $\widetilde t_k$ ($m$ qubits) from the top $m$ fractional bits of $P_k+2^{-m-1}$.
\State \emph{Clip.}  Compare with $0$ and $1/2$ to get $\check t_k=\min\{\max\{\widetilde t_k,0\},1/2\}$ ($m$ qubits).
\State \emph{Write.}  Controlled on $b_k=0$, copy $\check t_k$ to the output.  Controlled on $b_k=1$, apply a CNOT from the flag to the first fractional bit of the output, which writes $0.10\cdots0=1/2$.
\State \emph{Uncompute.}  Reverse steps 4, 3, 2 and 1, which only read the output register.  At most $2q+2w+3m+O(1)=O(m)$ qubits are held between steps.
\end{algorithmic}
\end{algorithm}

Every step applies the same gates for every $k$, and all ancillas return to $\ket0$ for every $k$.  By linearity, $O_\tau$ therefore acts on superpositions as
\begin{equation}
  O_\tau\sum_ka_k\ket k\ketzero m=\sum_ka_k\ket k\ket{\tau_k}.
  \label{eq:node-oracle-superposition}
\end{equation}

\begin{proof}[Proof of \Cref{prop:angle-oracle}]
The construction above realizes $O_\tau\ket k\ketzero m=\ket k\ket{\tau_k}$ with $\tau_k$ as in \Cref{eq:node-oracle-output}, and all its arithmetic is done at width $w=m+g$.  Before the final rounding there are three approximate stages.  The $\arccos$ evaluation has error at most $c_{\rm acos}2^{-w}$ by part~(i) of \Cref{thm:reversible-elementary-functions}, and multiplying by the stored constant $[1/(2\pi)]_w<1$ does not increase it.  The stored constant satisfies $|[1/(2\pi)]_w-1/(2\pi)|\le2^{-w}$, which contributes at most $\pi2^{-w}$ since $\theta_k\le\pi$.  Rounding the product to $w$ bits adds at most $2^{-w}$.  The total is at most $C2^{-w}$ with $C:=c_{\rm acos}+\pi+1$, so the constant $g:=\lceil\log_2(4C)\rceil$ keeps the error before the final rounding below $2^{-m-2}$.  The final rounding adds at most $2^{-m-1}$, and the sum is below $2^{-m}$.  The endpoint $k=0$ is exact, and clipping cannot increase the error.  The resources follow from \Cref{eq:reversible-elementary-resources} with $w=m+O(1)$, plus $O(q)$ gates for $X_k$.
\end{proof}

\begin{proof}[Proof of \Cref{lem:node-precision}]
By \Cref{eq:angle-error}, $\norm{t-\tau}_\infty\le2^{-m}\le2^{-3q/2}\eps/(24\pi)$.  Multiplying by $(2\pi/\sqrt3)N^{3/2}=(2\pi/\sqrt3)2^{3q/2}$ gives the first bound, and multiplying by $N=2^q$ gives $N\norm{t-\tau}_\infty\le2^{-q/2}\eps/(24\pi)<1/(2\pi)$.  Also $m\ge q+1$, as \Cref{prop:angle-oracle} requires.
\end{proof}

\subsection{Row access}
\label{app:row-oracle}

The bound $d_r=5$ counts nodes per grid point, so it does not depend on whether $M_\sigma$ or its transpose is used.  (Numerically, the largest multiplicity is $3$ for $N\le64$ and $4$ for $128\le N\le2^{18}$, both for the nodes $t$ and for the negative nodes $t^-$ of \Cref{thm:nupct-reduction}.)

\begin{proof}[Proof of \Cref{lem:quantized-multiplicity}]
Rows $\ell>N/2$ are empty because $\sigma_k\le N/2$.  If $\sigma_k=\ell$, then $|N\tau_k-\ell|\le1/2$, and $N\norm{t-\tau}_\infty<1/(2\pi)<1/4$ gives
\begin{equation*}
  |Nt_k-\ell|
  \le |Nt_k-N\tau_k|+|N\tau_k-\ell|
  <\frac34.
\end{equation*}
Hence $\theta_k=2\pi t_k\in[\Theta_{\ell,-},\Theta_{\ell,+}]$, and since $k=\frac N2(1+\cos\theta_k)$, the index $k$ lies in \Cref{eq:expanded-index-window}.  Because cosine is $1$-Lipschitz,
\begin{equation}
  b_\ell-a_\ell
  \le \frac N2\bigl(\Theta_{\ell,+}-\Theta_{\ell,-}\bigr)
  \le \frac{3\pi}{2}<5,
  \label{eq:expanded-window-length}
\end{equation}
so the window contains at most five integers.  The bound is attained, for example, at $N=8$ and $\ell=2$.

For monotonicity, $x_{k+1}-x_k=2/N$ and $|\tfrac{d}{dx}\arccos x|\ge1$, so $\theta_k-\theta_{k+1}\ge2/N$ and $t_k-t_{k+1}\ge1/(\pi N)$.  Hence
\[
  \tau_k-\tau_{k+1}\ge t_k-t_{k+1}-2\norm{t-\tau}_\infty>\frac1{\pi N}-\frac{2}{2\pi N}=0.
\]
So $\tau_k$ is decreasing in $k$.  Rounding is non-decreasing, so $\sigma_k=s_k^{\mathrm{unw}}=\nint(N\tau_k)$ is non-increasing in $k$.  The set of $k$ with $\sigma_k=\ell$ is therefore a set of consecutive integers.
\end{proof}

\begin{proof}[Proof of \Cref{prop:row-oracle}]
For rows $\ell>N/2$, which are empty, every slot returns the padding value $N+r$.  For $0\le\ell\le N/2$, compute $h_\ell=\lfloor\widetilde a_\ell\rfloor$ from an approximation with $|\widetilde a_\ell-a_\ell|<1/4$.  This needs $\cos\Theta_{\ell,+}$ to $q+O(1)$ bits (part~(ii) of \Cref{thm:reversible-elementary-functions}), since a cosine error below $1/(2N)$ gives an error below $1/4$ in $a_\ell$.  With $A_\ell=\lfloor a_\ell\rfloor$ we have $h_\ell\in\{A_\ell-1,A_\ell,A_\ell+1\}$, and by \Cref{lem:quantized-multiplicity} every column of row $\ell$ lies in $[A_\ell,A_\ell+5]$.  So every column lies among the eight candidates
\begin{equation}
  \mathcal J_\ell=
  \setof{h_\ell-1,h_\ell,\ldots,h_\ell+6}.
  \label{eq:eight-candidates}
\end{equation}

By \Cref{lem:quantized-multiplicity}, the columns of row $\ell$ are consecutive integers $f_\ell,f_\ell+1,\ldots,f_\ell+d_\ell-1$, all in $\mathcal J_\ell$.  Hold each candidate $k$ in a signed $(q+3)$-bit register, which covers the range $-2,\ldots,N+7$, and flag whether $0\le k<N$.  If not, query index $0$ instead, so that every oracle call is on a valid index.  Mark $k$ valid if the flag is set and $\sigma_k=\ell$.  Then $f_\ell$ is the smallest valid candidate and $d_\ell$ is the number of valid candidates.  For empty rows, including all $\ell>N/2$, set $f_\ell=d_\ell=0$.  Computing the marks, deriving $f_\ell$ and $d_\ell$, copying them out and uncomputing the marks gives
\[
  \ket\ell\ket0\ket0\mapsto\ket\ell\ket{f_\ell}\ket{d_\ell}
\]
with eight calls each to $O_\sigma$ and $O_\sigma^\dagger$.

Since the row is consecutive, \Cref{eq:row_oracle} becomes
\[
  \rho_{\ell,r}=
  \begin{cases}
    f_\ell+r,&r<d_\ell,\\
    N+r,&d_\ell\le r<d_r.
  \end{cases}
\]
Hold $r$ in a $(q+3)$-bit register.  Compute the flag $c=[r<d_\ell]$, add $f_\ell$ to $r$ controlled on $c=1$, and add $N$ to $r$ controlled on $c=0$.  The register now holds $\rho_{\ell,r}$.  Since $f_\ell+d_\ell-1<N$ and $N+r\ge N$, the flag equals $[\rho_{\ell,r}<N]$, so recomputing this comparison erases it:
\[
  \ket\ell\ket r\ket{f_\ell,d_\ell}\ket0
  \mapsto\ket\ell\ket r\ket{f_\ell,d_\ell}\ket c
  \mapsto\ket\ell\ket{\rho_{\ell,r}}\ket{f_\ell,d_\ell}\ket c
  \mapsto\ket\ell\ket{\rho_{\ell,r}}\ket{f_\ell,d_\ell}\ket0.
\]

Uncompute $f_\ell$ and $d_\ell$ by reversing the first step, which depends only on $\ell$.  This gives \Cref{eq:row_oracle} for $0\le r<d_r$, the only inputs the sparse block encoding queries.  The two passes use sixteen calls each to $O_\sigma$ and $O_\sigma^\dagger$.  Each call costs $\widetilde O(m^2)$ by \Cref{prop:angle-oracle}, and the candidate and shift arithmetic acts on $O(q)$-bit words, which gives the stated cost.
\end{proof}

\subsection{Negative nodes}
\label{app:negative-branch}

\begin{proof}[Proof of \Cref{lem:negative-branch}]
First, $t_k\ge t_{N-1}>1/(2N)$: indeed $\cos(\pi/N)>1-\pi^2/(2N^2)\ge1-2/N$ for $N\ge4$, so $\theta_{N-1}=\arccos(1-2/N)>\pi/N$.  Hence $\tau_k\ge t_k-\norm{t-\tau}_\infty>1/(2N)-1/(2\pi N)>0$, and $(-\tau_k)\bmod1=1-\tau_k$.  Then
\[
  N\tau_k^-=N-N\tau_k=N-s_k^{\mathrm{unw}}-\frac{z_k}{2}=s_k^-+\frac{z_k^-}{2}.
\]
Since $\sigma_k=s_k^{\mathrm{unw}}\le N/2$, we have $\sigma_k^-=\ell$ exactly when $\sigma_k=(-\ell)\bmod N$, which is the statement about rows.

\Cref{lem:routing-factorization}, and hence \Cref{thm:dyadic-nuqft}, use the nearest integer only through the identity $N\tau_k=s_k^{\mathrm{unw}}+z_k/2$ with integer $s_k^{\mathrm{unw}}$ and $z_k\in[-1,1]$.  Its reduction modulo $N$ is $\sigma_k$.  So they apply to $\tau^-$ with $(s^-,z^-,\sigma^-)$.  This split agrees with rounding $N\tau_k^-$ directly, except possibly at a tie, where either choice is valid.
\end{proof}

\subsection{Coefficient preparation}
\label{app:coefficient-prep}

Both preparations are built from the following standard construction.

\begin{theorem}[M\"ott\"onen et al.\ \cite{MVBS05}]
\label{thm:mottonen-state-preparation}
Let $D=2^n$ and $a\in\mathbb R_{\ge0}^{D}$ with $\sum_ja_j^2=1$.  There is an ancilla-free circuit $U_a$ with
\begin{equation}
  U_a\ketzero n=\sum_{j=0}^{D-1}a_j\ket j,
  \label{eq:mottonen-target-state}
\end{equation}
made of $n$ uniformly controlled $R_y$ rotations, which decompose into at most
\begin{equation}
  D-1\quad R_y\text{ rotations}
  \qquad\text{and}\qquad
  D-2\quad \mathrm{CNOT}\text{ gates}.
  \label{eq:mottonen-real-count}
\end{equation}
For a binary prefix $z$, let $P_z:=\sum_{j\text{ with prefix }z}a_j^2$.  The rotation that splits $z$ into its two children has angle
\begin{equation}
  \vartheta_z
  =2\arcsin\sqrt{\frac{P_{z1}}{P_z}}
  \label{eq:mottonen-angle}
\end{equation}
(and $\vartheta_z=0$ if $P_z=0$).  If each of the $S$ rotations is implemented with operator-norm error at most $\eta$, the circuit has error at most $S\eta$.
\end{theorem}

This is the nonnegative-amplitude case of \cite{MVBS05}: no phase-equalizing $R_z$ rotations are needed, and \Cref{eq:mottonen-angle} is their Eq.~(8).  The last claim follows by replacing the rotations one at a time.

\begin{proof}[Proof of \Cref{prop:coefficient-prep}]
For fixed $r$, pad the amplitudes in \Cref{eq:chi-r} to dimension $K'=2^{\lceil\log_2K\rceil}<2K$ and apply \Cref{thm:mottonen-state-preparation}.  This gives $O(K)$ gates.  For the controlled family, pad $r$ to $K'$ values as well.  At level $j$ of the cascade the angle depends on $r$ and on the $j-1$ bits already prepared, so the level is one uniformly controlled rotation with $\log_2K'+j-1$ controls, which decomposes into $K'2^{j-1}$ rotations and as many CNOTs.  Summing over $j$ gives $K'(K'-1)=O(K^2)$ gates.  By the last part of \Cref{thm:mottonen-state-preparation}, angle errors $O(2^{-b_\alpha})$ give an error $O(K^22^{-b_\alpha})$ in $\PREP_{\mathrm{LCU}}$, at most twice that in the PREP-SELECT-PREP$^\dagger$ product, and $O(\alpha K^22^{-b_\alpha})$ in the encoded matrix.  \Cref{eq:prep-precision} makes this $O(\eps_{\mathrm{prep}})$.
\end{proof}

\begin{proof}[Proof of \Cref{prop:outer-lcu-compilation}]
The normalizations are derived in \Cref{lem:assembled-factor-error} and \Cref{subsec:uniform-versus-weighted-outer-lcu}.  Here we only build the circuits.  Both outer states have nonnegative amplitudes ($K^{-1/2}$, or $\sqrt{\lambda_r/\Lambda}$), so \Cref{thm:mottonen-state-preparation} with $D=K'<2K$ prepares them with $O(K)$ gates.  The rotation in (a) is one uniformly controlled $R_y$ rotation on the outer index, which also decomposes into $O(K)$ gates.  With $O(K)$ rotations, angle errors $O(2^{-b_{\mathrm{out}}})$ give an error $O(K2^{-b_{\mathrm{out}}})$ in the selected block, hence $O(K\alpha_{\mathrm{unif}}2^{-b_{\mathrm{out}}})$ in the encoded matrix, since $\alpha_{\mathrm{wt}}\le\alpha_{\mathrm{unif}}$.
\end{proof}

\section{Why \texorpdfstring{$\kappa$}{kappa} cannot be bounded for the NUCT nodes}
\label{app:kappa-nuct}

The NUQFT cost \Cref{eq:nuqft-kappa-cost} depends on the parameter $\kappa$ of \Cref{eq:kappa}.  This appendix explains why it does not give an explicit complexity for the NUCT.  For the nodes \Cref{eq:angular-nodes} we could neither compute $\kappa$ in closed form nor bound it by a polynomial in $N$. We show that finding a good upper bound is difficult:\\
Let
\begin{equation}
  \delta_N:=1-\max_{0\le k<N}|y_k|
  \label{eq:app-delta}
\end{equation}
be the smallest distance from an offset $y_k=z(t_k)$ to the cell boundary $\pm1$.  The negative nodes $t_k^-$ have offsets $-y_k$, so they give the same $\delta_N$.  The point $y_k^*$ in \Cref{eq:kappa} lies between $y_k$ and $\widehat y_k$, and $|y_k-\widehat y_k|\le N2^{-m+1}$ by \Cref{eq:source-residual-perturbation}.  If $N2^{-m+1}\le\delta_N/2$, then $|y_k^*|\le1-\delta_N/2$, so $1-(y_k^*)^2\ge\delta_N/2$ and
\begin{equation}
  \kappa\le\sqrt{2/\delta_N}.
  \label{eq:app-kappa-delta}
\end{equation}
Without such a condition on $m$ the point $y_k^*$ can approach $\pm1$ and no bound holds.  So both $\kappa$ and the precision $m\ge q+2+\log_2(1/\delta_N)$ are controlled by $\delta_N$, and bounding $\kappa$ means bounding $\delta_N$ from below.

\paragraph{Exponential Bound.}  Since $\theta_k=2\pi t_k$ and $Nt_k=s_k^{\mathrm{unw}}+y_k/2$ by \Cref{eq:dyadic-routing-data-a}, we have $N\theta_k=2\pi s_k^{\mathrm{unw}}+\pi y_k$.  Since $T_N(\cos\theta)=\cos(N\theta)$, $T_N(x_k)=\cos(N\theta_k)=\cos(\pi y_k)$.  With $T_N=2T_{N/2}^2-1$ this gives
\[
  |T_{N/2}(x_k)|=|\cos(\pi y_k/2)|=\sin\bigl(\tfrac\pi2(1-|y_k|)\bigr).
\]
Since $u\le\sin(\pi u/2)\le\pi u/2$ for $u\in[0,1]$,
\begin{equation}
  \frac2\pi\min_{0\le k<N}|T_{N/2}(x_k)|
  \le\delta_N\le
  \min_{0\le k<N}|T_{N/2}(x_k)|.
  \label{eq:app-delta-chebyshev}
\end{equation}
So $\kappa$ is small exactly when no grid point $x_k=-1+2k/N$ comes close to a zero $\cos\bigl((2n+1)\pi/N\bigr)$ of $T_{N/2}$.\\

 The polynomial $T_{N/2}$ has integer coefficients and degree $N/2$, and $x_k=(2k-N)/N$ has denominator dividing $2^{q-1}$.  Hence $2^{(q-1)N/2}\,T_{N/2}(x_k)$ is an integer.  It is nonzero.  The zeros of $T_{N/2}$ are cosines of rational multiples of $\pi$, and by Niven's theorem \cite{Niv56} such a cosine is rational only if it lies in $\{0,\pm\frac12,\pm1\}$.  None of these is a zero of $T_{N/2}$ when $N=2^q\ge4$.  Therefore $|T_{N/2}(x_k)|\ge2^{-(q-1)N/2}$, and by \Cref{eq:app-delta-chebyshev,eq:app-kappa-delta}
\[
  \delta_N\ge\frac2\pi\,2^{-(q-1)N/2},
  \qquad
  \log_2\kappa\le\frac{(q-1)N}{4}+O(1).
\]
This is the best bound we obtained.  In \Cref{eq:nuqft-kappa-cost} it gives $\log(1+\kappa L)=O(N\log N)$ qubits and gates, which removes the exponential speedup.

We now give some heuristic arguments for why such a bound might be hard. A polynomial bound $\kappa\le\mathrm{poly}(N)$ needs every fractional part of $N\theta_k/(2\pi)=\frac{N}{2\pi}\arccos(-1+2k/N)$ to stay at distance $1/\mathrm{poly}(N)$ from $\frac12$, uniformly in $k$ and $N$.  Near $x=0$, for instance, $N\theta_k/(2\pi)\approx N/4-(k-N/2)/\pi$, so even the linearized problem asks how well multiples of $1/\pi$ approximate half-integers, which depends on the irrationality measure of $\pi$.  Away from $x=0$ the dependence on $k$ is nonlinear, and we know of no argument that controls these fractional parts.  An equidistribution heuristic suggests $\delta_N\approx1/N$ and $\kappa\approx\sqrt N$, but it gives no bound valid for every $N$.

\paragraph{Numerics.}  \Cref{tab:app-kappa} lists $\delta_N$ and $\kappa_0:=\max_k(1-y_k^2)^{-1/2}$, the limit of $\kappa$ as $m\to\infty$.  The values were computed in double precision and checked with $50$-digit arithmetic for $q\le12$ and for the minimizer at $q=20$.  The growth is roughly $\sqrt N$ but irregular.  The product $N\delta_N$ ranges from $0.055$ to $5.2$, $\kappa_0$ is not monotone in $N$, and at $q=20$ the single node $k=1010685$ has $\delta_N\approx5.2\times10^{-8}$.  Outliers of this kind are why no clean bound can be read off from the data either.

\begin{table}[h]
\centering
\begin{tabular}{@{}rrrrr@{}}
\toprule
$q$ & $\delta_N$ & $N\delta_N$ & $\kappa_0$ & $\kappa_0/\sqrt N$\\
\midrule
4  & $2.87\times10^{-1}$ & 4.59  & 1.43   & 0.36\\
6  & $4.94\times10^{-2}$ & 3.16  & 3.22   & 0.40\\
8  & $1.15\times10^{-2}$ & 2.94  & 6.62   & 0.41\\
10 & $3.86\times10^{-4}$ & 0.40  & 36.0   & 1.12\\
12 & $2.09\times10^{-4}$ & 0.86  & 48.9   & 0.76\\
14 & $3.20\times10^{-4}$ & 5.24  & 39.6   & 0.31\\
16 & $2.34\times10^{-5}$ & 1.54  & 146.1  & 0.57\\
18 & $5.26\times10^{-6}$ & 1.38  & 308.2  & 0.60\\
20 & $5.22\times10^{-8}$ & 0.055 & 3094.6 & 3.02\\
\bottomrule
\end{tabular}
\caption{Distance $\delta_N$ of the NUCT offsets to the cell boundary and the resulting $\kappa_0$, for $N=2^q$.}
\label{tab:app-kappa}
\end{table}

\paragraph{The fix in our algorithm.}  \Cref{thm:dyadic-nuqft} does not need $\delta_N$.  The $\arccos$ evaluations inside the NUQFT act on exact inputs (\Cref{lem:diagonal-factor-implementation}), and the node error enters only through the Fourier phases, whose derivative has no singularity.  This is why the NUCT complexity in \Cref{sec:end-to-end} is explicit.

\section{Uniform and weighted outer LCU constructions}
\label{subsec:uniform-versus-weighted-outer-lcu}

\Cref{lem:assembled-factor-error} combines the $K$ component block
encodings of $A_r=D_{v_r}FM_\sigma D_{u_r}$, whose normalizations are
$\beta_r=\sqrt{d_r}\lambda_r$, through one of two outer LCUs.  The
normalization-aware (weighted) outer LCU prepares a selector with amplitudes
$\sqrt{\lambda_r/\Lambda}$, so each branch weight cancels its own
$\lambda_r$ and the normalization is $\sqrt{d_r}\Lambda=O(\sqrt{d_r})$.
The uniform outer LCU first pads every branch from normalization $\beta_r$ to the
common value $\beta_\star=\sqrt{d_r}\Lambda\ge\beta_r$ by an ancillary one-qubit amplitude, and then prepares $\ket{+_K}$.  It selects $\frac1K\sum_r\widetilde A_r/\beta_\star$, so it costs an
additional factor $K$ and gives normalization $K\sqrt{d_r}\Lambda=O(K\sqrt{d_r})$, with the same unnormalized error \Cref{eq:weighted-assembly-error}.
The proof of \Cref{thm:dyadic-nuqft} uses only the weighted alternative.  The
uniform outer LCU, which is close to the outer-LCU organization of
\cite{AKY26}, is retained only for comparison, and the two constructions are
never applied simultaneously.

The two outer-LCU normalizations, compiled explicitly in
\Cref{prop:outer-lcu-compilation}, obey
\[
\begin{array}{c|cc}
 & \Lambda=O(K)
 & \text{\Cref{lem:constant-bessel-mass}: }\Lambda=O(1)\\
\hline
\text{uniform outer LCU}
 & O(K^2\sqrt{d_r})
 & O(K\sqrt{d_r})\\
\text{normalization-aware outer LCU}
 & O(K\sqrt{d_r})
 & O(\sqrt{d_r})
\end{array}
\]
Relative to the source-level bound $O(K^2\sqrt{d_r})$, the constant-mass
estimate of \Cref{lem:constant-bessel-mass} and the weighted selector each remove
one factor of $K$, and the off-diagonal entries of the table show that the
two refinements act independently.  The constant-mass estimate is not needed
to remove $\kappa$: with the source-level bound $\Lambda=O(K)$, the
choice $p=O(\log(\sqrt{d_r}K^2/\eps))=O(q+\log(1/\eps))$ still suffices in
\Cref{thm:dyadic-nuqft}.  Its role is to improve the normalization, and hence
the success probability.  For the NUCT, with $d_r\le5$, the bounds become $O(K^2)$ for the
source-level construction, $O(K)$ for the uniform construction, and $O(1)$
for the weighted construction.  These are comparisons between the upper
bounds proved for the respective constructions, not a lower bound showing
that the original NUQFT architecture cannot be improved by incorporating the
same refinements.

The normalization does not change the gate count of a single transform-layer
invocation.  However, for an $\alpha$-normalized block encoding of an
operator $\widetilde F$ applied to an input proportional to $f$, the
postselection success probability is
$\norm{\widetilde Ff}_2^2/(\alpha^2\norm{f}_2^2)$.  Relative to the uniform
construction, the weighted construction therefore improves this probability
by a factor $K^2$ and reduces the amplitude-amplification overhead by a
factor $K=O(q+\log(1/\eps))$, which carries over to coherent output-state
preparation and to applications that repeatedly invoke the transform block
encoding.  The selector states and the normalization-equalizing rotations
are compiled in \Cref{prop:outer-lcu-compilation}.

\end{document}